\documentclass[11pt]{article}

\usepackage{booktabs} 
\usepackage{amsmath}
\usepackage{amssymb}
\usepackage{amsthm}
\usepackage[usenames, dvipsnames]{color}
\usepackage{url}
\usepackage{mathrsfs}
\usepackage{amsfonts}
\usepackage{graphicx}

\usepackage[caption=true]{subfig}
\usepackage{algorithm}
\usepackage{algpseudocode}

\usepackage[margin=1in]{geometry}

\newtheorem{dfn}{Definition}
\newtheorem*{dfn*}{Definition}
\newtheorem{lem}{Lemma}
\newtheorem{thm}{Theorem}
\newtheorem*{thm*}{Theorem}
\newtheorem{claim}{Claim}

\newtheorem{cor}{Corollary}

\begin{document}
\title{A Descending Price Auction for Matching Markets}
\author{Shih-Tang Su \and Jacob D. Abernethy \and Grant Schoenebeck \and Vijay G. Subramanian}
\date{}

\maketitle

\begin{abstract}
This work presents a descending-price-auction algorithm to obtain the maximum market-clearing price vector (MCP) in unit-demand matching markets with $m$ items by exploiting the combinatorial structure. With a shrewd choice of goods for which the prices are reduced in each step, the algorithm only uses the combinatorial structure, which avoids solving LPs and enjoys a strongly polynomial runtime of $O(m^4)$. Critical to the algorithm is determining the set of under-demanded goods for which we reduce the prices simultaneously in each step of the algorithm. This we accomplish by choosing the subset of goods that maximize a skewness function, which makes the bipartite graph series converges to the combinatorial structure at the maximum MCP in $O(m^2)$ steps. A graph coloring algorithm is proposed to find the set of goods with the maximal skewness value that yields $O(m^4)$ complexity.
\end{abstract}
\newpage
\section{Introduction} \label{sec:intro}
	Online Advertising is an over \$70 billion business
 with double-digit growth in consecutive years over a period of many years. Since nearly all of the ads are sold via auction mechanisms, auction-based algorithm design, which focuses on the online advertising, has become an important class of mechanism design to study. Among all online advertising auctions, the sponsored search auction, also known as a keyword auction, is the one that most captures researchers' attention.

	In a typical sponsored search auction, the auctioneer has a set of web slots to  sell and every advertiser has different valuations on different web slots. Problems in sponsored search auctions are usually modeled as problems in (cardinal-preference) matching markets, and prices are used to clear the market. The key assumption of sponsored search auctions is that every advertiser shares the identical ordinal preference on web slots.
	Under this assumption, the celebrated Vickrey-Clarke-Groves (VCG) mechanism \cite{VCG_V,VCG_C,VCG_G}, which makes truthful bidding by the advertisers as (weakly) dominant strategies but yields low revenue to the auctioneer, is adopted by some web giants such as Facebook\footnote{Facebook Ad Auction: See https://www.facebook.com/business/help/163066663757985.}, and is also a robust option in scenarios where the revenue equivalence theorem~\cite{krishna} holds. However, as the revenue equivalence theorem does not hold in multi-good auctions~\cite{bayes}, auctioneers can look for greater expected revenue than the value obtained by VCG mechanism by using even different efficient and market-clearing auction mechanisms. The most popular auction among these mechanisms is the Generalized Second Price (GSP) auction employed by Google. Since the GSP is not incentive compatible, its equilibrium behavior needs to be analyzed  \cite{bayes,caragiannis2015bounding,edelmanstrategic}, and there are some Bayesian Nash equilibria (BNE) \cite{krishna} that have greater expected revenue than the expected revenue of the VCG mechanism. It should also be noted that designing (revenue) optimal mechanisms \cite{myerson1981optimal} is intractable~\cite{cai2012optimal,daskalakis2014complexity} even in the context of matching markets when there is more than one good. Thus, the possibility of higher expected revenue coupled with the ease of implementing the GSP auction and the intractability of optimal mechanisms has lead to the popularity of the GSP mechanism.

    Unlike a decade ago where there were only statically-listed ads, websites now serve a variety of ads simultaneously, including sidebar images, pop-ups, embedded animations, product recommendations, etc. With this in mind, and the growing heterogeneity in both advertisers and consumers, it is clear that the ``shared ordinal preference'' assumption is untenable in the context of market design. Search engines and ad-serving platforms will be faced with a growing need to implement general unit-demand matching markets~\cite{gale1962}, and such market settings are the focus of our work.

    We refer to the prices that efficiently allocate the set of goods to the bidders according to their private valuations as a vector of \emph{market-clearing prices} (MCP). An ascending price auction algorithm that generalizes the English auction was presented by Demange, Gale and Sotomayor \cite{constructMCP}. This ascending price algorithm (DGS algorithm) obtains the element-wise minimum 
MCP, that coincides with the VCG price. DGS is thus incentive compatible yet obtains low expected revenue for the mechanism. Of course, simultaneously maximizing revenue and maintaining incentive compatibility is computationally intractable once we have more than one good for sale, but we should still hope to obtain better than the \emph{minimum} MCP within efficient mechanisms.

    In the present paper we design a family of mechanisms that seek to elicit the \emph{maximum} MCP from the participants without sacrificing computational efficiency. Here we focus explicitly on how we can efficiently compute the maximum MCP given some representation of the bidder utilities, and defer the general\footnote{Several illustrative instances of Bayesian equilibrium of strategic buyers are discussed in the full version \cite{fullver}. One particular instance explicitly demonstrates an example where our mechanism yields greater expected revenue when compared to the expected revenue of the VCG mechanism.} analysis of strategic behavior to future work. 

    Critical to our paper would be answering whether there exists a strongly polynomial-time algorithm to obtain the maximum MCP exactly. Before discussing the literature for the maximum MCP, we discuss the state of the art for the minimum MCP.  The intuitively appealing DGS ascending price algorithm that attains the minimum MCP is only weakly polynomial time: the potential function used to show convergence depends on the valuations; and its value decrements by at least a constant independent of the valuations in each step. In fact, it is the well-known strongly polynomial-time Hungarian algorithm \cite{kuhn1955} for finding the maximum weight matching in a weighted bipartite graph, that yields a strongly polynomial time algorithm for finding the minimum MCP, $O(m^4)$ in the original implementation that can then be reduced to $O(m^3)$~\cite{edmonds1972theoretical}. This will be the aspirational goal of this work. 
    
    Using the method outlined in \cite{Truth-LP}, where one computes the solution of two linear programs, it is possible to determine the maximum MCP. Note, however, that this is at best a weakly polynomial-time algorithm, and is neither a combinatorial nor an auction algorithm. Given that the DGS ascending price mechanism returns the minimum MCP, it is also intuitive to study descending price mechanisms to obtain the maximum MCP, i.e., generalize the Dutch auction to multiple goods. 
    The first attempt to obtain the maximum MCP through descending price auction is in the work by Mishra and Garg \cite{mishragarg}, where they provide a descending-price-based auction algorithm that yields an approximation algorithm. The algorithm doesn't require agents to bid their whole valuation but still yields a price-vector in weakly polynomial-time\footnote{Again, as the number of iterations depends on both $\epsilon$ and the input valuation matrix.} that is within $\epsilon$ in $l_\infty$ norm of the maximum MCP\footnote{Even though the final price may not be market clearing, decreasing it further by $\epsilon$ and then running the DGS algorithm, it is possible to obtain a market-clearing price that is within $2\epsilon$ of the maximum MCP.}. 
    Therefore, in this work, one of main goals is to develop a strongly polynomial-time combinatorial/auction algorithm using descending prices for the \textit{exact} computation of the maximum MCP. Note that based on the analysis in \cite{position} choosing the maximum MCP in sponsored search markets has exactly the same complexity as the VCG, GSP and Generalized first-price (GFP) mechanisms: The web-slots are sold from best to the worst and in decreasing order of the bids of the agents, with the only difference being the price that's ascribed to each good. Once the computational problem is solved, setting the maximum MCP is a viable option for general unit-demand markets, and is an alternate efficient mechanism.

\subsection{Our contribution}

        By judiciously exploiting the combinatorial structure in matching markets, we propose a strongly polynomial-time\footnote{See Section~\ref{sec:prelims} for a definition of strongly and weakly polynomial-time complexity.} descending price auction algorithm that obtains the maximum MCPs in time $O(m^4)$ with $m$ goods (and bidders). Critical to the algorithm is determining the set of under-demanded goods (to be defined precisely later on) for which we reduce the prices simultaneously in each step of the algorithm. This we accomplish by choosing the subset of goods that maximize a skewness function, which is obtained by proposed graph coloring algorithm a simple combinatorial algorithm to keep updating the bipartite graph and the collection under-demanded goods set. We start by discussing an intuitively appealing algorithm to solve this problem that uses the Hopcraft-Karp~\cite{hopcroft} algorithm and Breadth-First-Search (BFS). This procedure will only yield a complexity of $O(m^{4.5})$. We will then present a refinement that cleverly exploits past computations and the structure of the problem to reduce the complexity to $O(m^4)$. 



\section{Related Work} \label{sec:related}

	While sponsored search auctions are a recent motivation to study matching markets, there is a vast history to the problem. The term ``matching market" can be traced back to the seminal paper ``College Admissions and the Stability of Marriage" work by Gale and Shapley \cite{gale1962}. In matching markets, the necessary and sufficient condition for the existence of an efficient matching using Hall's marriage theorem~\cite{hall} has been proved \cite{gale1960theory} and a widely used mathematical model of two-sided matching markets was introduced in ``The Assignment Game I: The Core" \cite{core} by Shapley and Shubik. In \cite{core}, the set of MCPs is further shown to be solutions of a linear programming (LP) problem and the lattice property is also established. Despite this the study of this problem goes back at least to the well-known strongly polynomial-time Hungarian algorithm \cite{kuhn1955} for finding the maximum weight matching in a weighted bipartite graph, which in fact can also be used to find the minimum MCP. Furthermore, several auction algorithms enhancing the run-time efficiency in markets with specific properties have been presented in \cite{bertsekas1992forward,bertsekas1993reverse}.
    Leonard \cite{Truth-LP} considered mechanisms with sealed-bids and proved that charging the minimum competitive equilibrium price from bidders will result in an incentive compatible mechanism, and also that MCP coincides with the VCG price. Soon after, an ascending-price-based auction~\cite{krishna} algorithm was presented by Demange, Gale, and Sotomayor (DGS) in \cite{constructMCP}, which starts at the zero-price vector and then increases the posted price for any of the minimal over-demanded sets \cite{gale1960theory} of goods to obtain the minimum MCP. Thereafter, plenty of ascending-price-based auction mechanisms have been studied under different assumptions in 
\cite{bikhchandani2006, gul2000english, ausubel2004}. We pause here to remind the reader that the DGS ascending price algorithm is only known to be weakly polynomial-time.
    
    On the other hand, there has only been a limited study of descending-price auction algorithms to obtain the maximum MCP. Mishra and Parkes present a descending price auction called the Vickrey-Dutch auction to generate the VCG price in equilibrium \cite{mishramulti}. To aim for a higher revenue for sellers, Mishra and Garg generalized the Dutch auction to provide a descending-price-based approximation algorithm in \cite{mishragarg}. As mentioned in Section \ref{sec:intro}, Mishra and Garg's algorithm yields an approximation to the maximum MCP via a weak polynomial-time algorithm, and furthermore, there is no analysis of the strategic bidding in their work. We remark again that the sequential LP approach in \cite{Truth-LP} can be used to obtain the maximum MCP via a weakly polynomial-time algorithm.
    
    Finally, there is a body of literature that attempts to raise the revenue of sellers in equilibrium in related problems, such as sponsored search auctions and combinatorial auctions. In sponsored search markets, Ghosh and Sayedi considered a two-dimensional bid on advertisers' valuations according to exclusive and nonexclusive display \cite{ghosh2010expressive}, then run a GSP-like auction to determine the allocation that maximizes the search engine's revenue. With this small variation, efficiency does not hold for GSP, and hence the revenue will be different from the VCG mechanism. Additionally, in combinatorial auctions, it is well-known that designing a revenue maximizing auction mechanism is still an open problem. To achieve a higher expected revenue of sellers, Likhodedov and Sandholm presented a class of auctions, called virtual valuations combinatorial auctions \cite{likhodedov2004methods}, to maximize the sum of a pre-determined weighted valuation and an evaluation function of allocation rather than maximizing the total valuations as in the VCG mechanism to get a higher revenue.

\section{Preliminaries and Problem Formulation} \label{sec:prelims} \label{sec2}

\paragraph{Bipartite Graphs}
        A \textbf{bipartite graph} $G=(\mathcal{M},\mathcal{B},E)$  is a graph such that the vertices $\mathcal{M}\cup \mathcal{B}$ can be divided into two disjoint subsets, $\mathcal{M}$ and $\mathcal{B}$, and there are no edges connecting vertices in the same subset, $E\subseteq \mathcal{M}\times \mathcal{B}$.  Such a graph is \textbf{balanced} if  $|\mathcal{M}| = |\mathcal{B}|$, i.e., if the two subsets have the same cardinality.  A \textbf{perfect matching} in a balanced bipartite graph  $G=(\mathcal{M},\mathcal{B},E)$ is a subset of edges $E_{pm} \subseteq E$ such that every vertex in $G$ is incident upon exactly one edge of the matching. We denote the neighbors of a set of vertices $S$ by $N(S)$, where $N(S)\triangleq \{j\in \mathcal{B}: \exists i \in S~\text{s.t.}~(i,j) \in E\}$ when $ S \subseteq \mathcal{M}$ and $N(S) \triangleq \{j\in \mathcal{M}: \exists i \in S~\text{s.t.}~(j,i) \in E\}$ when $ S \subseteq \mathcal{B}$.

    \begin{dfn}
        A set $S \subseteq \mathcal{M}$ or $S \subseteq \mathcal{B}$ in a bipartite graph $G=(\mathcal{M},\mathcal{B},E)$ is called a \textbf{constricted set} if $|S|>|N(S)|$. More precisely, we call $S$ a \textbf{constricted good set} if $S \subset \mathcal{M}$ or a \textbf{constricted buyer set} if $S \subset \mathcal{B}$.
    \end{dfn}

\begin{thm*}[\textbf{Hall's marriage theorem}~\cite{hall}]  For a balanced bipartite graph $G=(\mathcal{M},\mathcal{B},E)$, $G$ contains no perfect matching if and only if $G$ contains a constricted set.
\end{thm*}

\paragraph{Matching Market} We consider a matching market with a set $\mathcal{B}$ of buyers, and a set $\mathcal{M}$ heterogeneous merchandise with exactly one copy of each type of good. Each buyer $i\in \mathcal{B}$ has a non-negative valuation $v_{ij}\geq 0$ for good $j\in M$, and desires at most one good (e.g. they are unit-demand buyers).  We denote the $|\mathcal{B}|\times |\mathcal{M}|$ valuation matrix by $\mathbf{V}$. 
   Our assumption that  $|\mathcal{B}|=|\mathcal{M}|=m$ is without loss of generality because we can always add dummy goods or dummy buyers for balance. 

Given a price vector $\mathbf{P}=[P_1 \, P_2 \, ...\, P_m]$, we assume a  quasi-linear utilities for the buyers, i.e., buyer $i$ receiving good $j$ has utility $U_{i,j}=v_{i,j}-P_j$. Since each buyer is unit-demand, we define $U^*_i$ be the maximum (non-negative) payoff of buyer $i \in \mathcal{B}$, i.e., $U^*_i=\max \big\{0,\underset{j\in \mathcal{M}}{\max}~v_{i,j}-P_j\big\}$. Since buyers can opt out of the market and obtain zero, we insist on the payoff being non-negative. 

\begin{dfn} \label{def:pgset}
    Under a price vector $\mathbf{P}$, the \textbf{preferred-good set} of buyer $i\in \mathcal{B}$ is a set of goods $L_i\subseteq \mathcal{M}$ such that getting each good in $L_i$ maximizes buyer $i$'s payoff, $L_i=\{j\in \mathcal{M}|v_{i,j}-P_j=U^*_i\}$.
\end{dfn}
\noindent Note that the preferred goods set of a buyer is empty if its payoff for all the goods is negative.  

\begin{dfn} By connecting each buyer with its preferred goods and recalling the assumption of $|M|=|B|$, we can construct a balanced bipartite graph which we call the \textbf{preference graph}, i.e., $G_{\mathrm{pref}}=(\mathcal{M},\mathcal{B},E_{\mathrm{pref}})$ where $E_{\mathrm{pref}}=\{(j,i): i \in \mathcal{B}$ and $j\in L_i\}$.
\end{dfn}
\noindent To avoid any confusion, we always place goods on the left-hand side and buyers on the right-hand side of the preference graphs.

\begin{dfn} The set of goods $M$ is over-demanded in $G_{\mathrm{pref}}$ if it's a union of preferred-good sets of a set of buyers $B$, where $|B|>|M|$. Given a particular preference graph $G_{\mathrm{pref}}$ that doesn't contain a perfect matching and a constricted buyer set $B$, an over-demanded set of goods coincides with the neighbor set of $B$, i.e., $N(B)$, where the neighbor set is determined in $G_{\mathrm{pref}}$. Similarly, the under-demanded set of goods $M$ coincides with a constricted good set. 
\end{dfn}

Given a specific price vector, if the preference graph $G_{\mathrm{pref}}$ contains a perfect matching $E_{pm}\subseteq E_{\mathrm{pref}}$, then we can allocate to each buyer exactly one of the goods it prefers and also sell all the goods. A price vector that leads to a perfect matching in the realized preference graph is called a \textbf{market-clearing price (MCP)} (also called a Walrasian price).

Given any valuation $V$, it is well known that the set of MCPs is non-empty and bounded \cite{core}. Boundedness is obvious from the finiteness of the valuations. Non-emptiness is established either using the characterization in \cite{core}, using a constructive ascending price algorithm \cite{constructMCP} that starts from all the prices being $0$, or by using the VCG mechanism price (see Chap 15 in \cite{textbook}). Furthermore, the set of MCPs has a lattice structure \cite{core}, so that given any two different MCP vectors, the element-wise maximum of the vectors and the element-wise minimum of the vectors are also MCPs. This guarantees the existence of the maximum and the minimum MCPs. 

\paragraph{Complexity of Algorithms}
An algorithm runs in \textbf{strongly polynomial time} if the number of operations and the space used are bounded by a polynomial in the number of input parameters, i.e., $O(\text{polynomial of }|M|)$, but both do not depend on the size of the parameters (assuming unit time for basic mathematical operations). If this does not hold but the number of operations is still bounded by a polynomial in the number of input parameters where the coefficients depend on the size of the parameters, then we say that the algorithm runs in \textbf{weakly polynomial time}.

\section{Design of Descending Price Algorithm}  \label{sec3}
The problem considered in our work, as mentioned earlier, is to find the  generalization of the Dutch auction\footnote{Despite a similar sounding name, what we seek to implement is completely different from the generalized first price (GFP) auction as in \cite{hoy2013dynamic, Expressiveness}.} to matching markets.  Specifically, we seek a descending price auction that always converges to the maximum MCP. Like the DGS mechanism, our mechanism will choose a particular constricted good set to ensure the convergence. Specifically, we will define a dual to the ``minimal over-demanded set" which we call the \textit{maximally skewed set}.  Unlike minimal over-demanded sets, the maximally skewed set is unique, and an example of failure to achieve the maximum MCP if this set is not chosen will be discussed in Section \ref{sec6.1}.


\subsection{Framework of Descending Price Algorithms}~\label{sec:framework}

	
	We design a descending auction, which is the analogue of the ascending auction, in a straightforward framework. We start from a high enough initial price, iteratively pick a constricted good set to decrement prices, and terminate the algorithm when there exists a perfect matching. Clearly, this framework does not guarantee the termination in finite time, let alone strongly polynomial time. In order to make the algorithm run in strongly polynomial time, we will exploit the combinatorial structure of the preference graph, and make the evolution of the preference graph in the run of the algorithm be such that any specific bipartite graph appears at most once. To achieve this goal, we will specify a particular initial configuration, and a particular price reduction to be carried out in each step of the algorithm.
	
\paragraph{Initial Price Choice:} A perfect matching requires every good be preferred by some buyers. Then a reasonable starting point should guarantee that the preferred-buyer set of every good is non-empty, otherwise it cannot be an MCP for any valuation matrix. Thus, the natural candidate for the initial price is $P_j =\max_{i \in \mathcal{M}}v_{i,j}$ for good $j$, which is (element-wise) greater than or equal to any MCP but ensures that every good is preferred by at least one buyer from the very outset.
		
\paragraph{Price Reduction:}
In computing the price reduction for a given constricted set $S$, we need to reduce the price  by a large enough amount to trigger a change in the preference graph (otherwise we still have a constricting set and the same set of goods can be chosen again), but we should also avoid reducing the price of any good below its price in the maximum MCP. In other words, we want to find the minimum value to compensate the buyers not in the $N(S)$ to make at least one buyer indifferent between one of the goods in $S$ and the good(s) she prefers initially; the buyer in question may have an empty preferred goods set, in which case it is sufficient to ensure that one of the goods in $S$ has a non-negative utility with the price reduction. Lemma \ref{lemA} formally states the price reduction to be used in the proposed family of descending price algorithms.
	
        \begin{lem} \label{lemA}
            Given a constricted good set $S$ and a price vector $\mathbf{P}$, the minimum price reduction of all goods in this set $S$ guaranteeing to add at least a new buyer to the set $N(S)$ is            
\begin{equation}
\underset{i \in B \setminus N(S),l\in S}{\min} \{\underset{k \in M \setminus S}{\max}(v_{i,k}-P_k)- (v_{i,l}-P_l) \}.
\end{equation}
        \end{lem}

        \begin{figure}[t]
	        \centering
           \includegraphics[width=0.8\textwidth,height=1in]{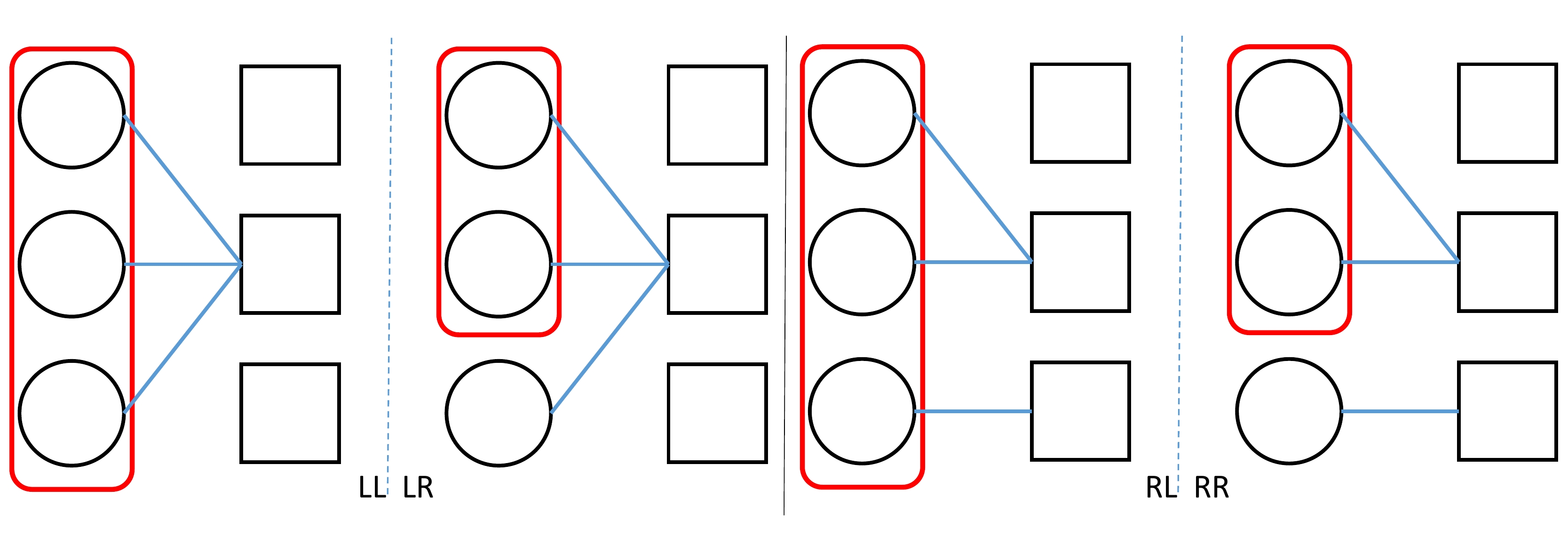}
           
\vspace{-12pt}\caption{Criteria for choosing constricted good sets}
            \label{fig:skewness}
        \end{figure}
        
\subsection{Choice of Constricted Good Sets and The Skewness Function} \label{skewfunc}

       Since different choices of constricted good sets could generate different MCPs when the algorithm terminates, pinpointing the right constricted good sets iteratively has a pivotal role when designing the algorithm for finding the maximum MCP\footnote{Appendix \ref{sec6.1} provides an example where a different choice fails to obtain the maximum MCP.}. Before detailing the selection criterion, we use Figure~\ref{fig:skewness} to provide some quick intuition. On one hand, we prefer choosing the constricted good set in LL to LR because we want to choose the largest good set given the same set of neighbors (buyers). On the other hand, we prefer RL to RR because we do not include any subgraph (set of good-buyer pairs) that already has a perfect matching. With this intuition in mind, we present the following formal criteria for choosing constricted-good sets:
            \begin{enumerate}
                \item Pick the constricted goods sets $S$ with the largest difference $|S|-|N(S)|$.
                \item If there are multiple sets with the same $|S|-|N(S)|$, choose the one with the smallest size.
            \end{enumerate}
            The first criterion ensures that at each step the algorithm (simultaneously) reduces the price of the most critical set of goods. The proof that our algorithm returns the maximum MCP will not hold without this property. As an added bonus it also positively impacts the speed of convergence. The second criterion excludes any subset of goods $S'\subset S$ which is already perfectly matched to a subset of buyers, i.e., $|N(S')\setminus N(S)|\geq |S'|$.
Jointly the criteria imply that we are searching for the most ``skewed" constricted good set in the preference graph. To formulate this mathematically, we define a function to measure the skewness of a set.
            \begin{dfn} \label{dfn:skew}
            The skewness of a set of goods $S$ is defined by function $f:2^\mathcal{M}\setminus {\emptyset} \mapsto R$ with $f(S)=|S|-|N(S)|+\frac{1}{|S|}$ for all $S\subseteq \mathcal{M}$ with $S\neq \emptyset$, where $2^\mathcal{M}$ is the power set of $\mathcal{M}$.
            \end{dfn}

            With this skewness function, the criteria described earlier are equivalent to choosing the constricted goods set with the maximal skewness. 
            To formally make this statement we need to show two properties. The first one is the uniqueness of the maximally skewed set when the preference graph has no perfect matching; and the second one is that the maximally skewed set is a constricted goods set when the preference graph has no perfect matching. Lemma~\ref{lem2} proves these. 

            \begin{lem} \label{lem2}
                Given a bipartite graph with no perfect matching, the maximally skewed set is unique and coincides with the constricted goods set with the maximal skewness.
            \end{lem}
	
	        With Lemma \ref{lem2} in place, it easily follows that the two rules we imposed are equivalent to finding the maximally skewed set at every iteration (as we already know that a perfect matching doesn't exist). With the proper initial price vector choice, specified price reduction per round, and the unique choice of the maximally skewed set, the complete algorithm is described in Algorithm \ref{alg:alg1}.

		Note that the DGS algorithm, which searches for over-demanded sets to increase the price, has a dual structure to our algorithm. Thus, it is not surprising that the minimally over-demanded sets of items in the DGS algorithm, denoted as DGS sets below, have a relationship with the skewness function $f(\cdot)$. They are ones that obtain the minimum positive value of the function $|N(S)|-|S|-\frac{1}{|N(S)|}$ when the algorithm starts with initial price 0. We highlight the fact that the DGS sets may not be unique as there can be multiple sets of goods that yield the same minimum positive value for the function $|N(S)|-|S|-\frac{1}{|N(S)|}$. In contrast to our algorithm, the lack of uniqueness in the DGS algorithm is not as critical because different choices of DGS sets lead to the same minimum MCP. Understanding this contrast better is for future work.
        \begin{algorithm}[h]
            \caption{Skewed-set Aided Descending Price Auction}
            \begin{algorithmic}[1]
            \Require
            A $|\mathcal{B}| \times |\mathcal{M}|$ valuation matrix $\mathbf{V}$.
            \Ensure
            MCP $\mathbf{P}$.
            \State Initialization, set the price of good $j$, $P_j=\max_{i \in \mathcal{B}} v_{i,j}$.
            \State Construct the preference graph.
            \While{There exists a constricted good set}
                \State Find the maximally skewed set $\mathbf{S}$.
                \State For all $j \in \mathbf{S}$, reduce $P_j$ by                 $\min_{i \in \mathcal{B} \setminus N(\mathbf{S}),l\in \mathbf{S}} \{\max_{k \in \mathcal{M} \setminus     \mathbf{S}}(v_{i,k}-P_k)- (v_{i,l}-P_l) \}$.
                \State Construct the preference graph.
            \EndWhile
            \State Return $\mathbf{P}$.
            \end{algorithmic}
            \label{alg:alg1}
          
        \end{algorithm}
 \vspace{-12pt}
       \section{Price Attained, Convergence Rate, and Complexity}\label{correctness}

First, we demonstrate that the proposed skewed-set aided descending price auction algorithm returns the maximum MCP. We achieve this by performing a check by adding a fictitious dummy good to the preference bipartite graph at termination. Second, we use the potential function to prove the finite time convergence of the algorithm. Finally, we analyze the complexity of the algorithm by presenting algorithms to find the maximally skewed set.

	\subsection{Attaining Maximum Market-Clearing Price}
    In advance of analyzing the relationship between the skew-aided algorithm and the maximum MCPs, we have to precisely characterize the extremal nature of the maximum MCP. Wearing an optimization hat and using the idea of feasible directions, one would expect that checking whether the MCP of any good can be increased or not is straightforward\footnote{There is a history of such variational characterizations in the stable matching literature \cite{immorlica2005marriage,hatfield2005matching}
where agents are assumed to have ordinal preferences.}. However, this logic misses the underlying matching problem and the Marriage theorem. Additionally, since the skewed-set aided algorithm is built on the combinatorial structure of the problem, to bridge the maximum MCP to our algorithm requires a combinatorial characterization of the maximum MCP. The combinatorial characterization requires adding a fictitious dummy good to preserve the property. Hence, we have to provide the following definition before stating the variational and combinatorial characterizations of the maximum MCP in Theorem \ref{lem4.1}.

    \begin{dfn}Given a bipartite graph $G$, let $N_G^D(B)$ be the neighbor set of the buyer set $B$ after adding a dummy good---a good for which every buyer has value 0. If the graph $G$ is clear from the context, we will also simplify the notation further to $N^D(B)$ after adding a dummy good.
\end{dfn}

  \begin{thm} \label{lem4.1}
  An MCP $P^*$ is the maximum if and only if for any subset of goods, increasing the price of all goods in the set will change the preference graph such that no perfect matching exists. Equivalently, by adding a dummy good, $P^*$ is the maximum if and only if any subset of buyers $B$ has a cardinality strictly less than the cardinality of the set of buyers' neighbors $N^D(B)$.
  \end{thm}
  For further clarification, any buyer who has zero surplus at the maximum MCP (by definition of the maximum there will exist at least one such buyer) will be indifferent between the matched good and the dummy good. Hence, for every buyer set $B$ containing a zero-surplus buyer, the dummy good $D$ will be in the neighbor set of this set, $D\in N^D(B)$. We also show a dual property to VCG prices of the maximum MCP in Section~\ref{sec:externality}.

With Theorem \ref{lem4.1} in hand, we will now establish the correctness of the algorithm, assuming that it halts (in finite-time).
  	Since the skew-aided algorithm continually changes the preference graph, it is necessary to label the bipartite graph in each round of our algorithm before starting any analysis. Let $G_0$ be the initial bipartite graph, in the running of our algorithm, we obtain a bipartite graph $G_t$ at $t^{th}$ round. Then, we'll need to check whether the terminal condition holds. To avoid cumbersome notation, we will use $N_t^D(B)$ instead of $N_{G_t}^D(B)$.

	With Theorem \ref{lem4.1}, the proof of Theorem \ref{thm2} followings by checking that the preference graphs at termination coincides has the combinatorial characterization outlined above.
        \begin{thm} \label{thm2}
           The skewed-set aided descending-price algorithm always returns the maximum MCP.
        \end{thm}
   \subsection{Preference Graphs Converge Quadratically in the Number of Goods} \label{sec4.2}
        The Algorithm \ref{alg:alg1} changes the preference graph in each round to obtain the bipartite graph with combinatorial structure of MCP at termination. We will now show that the algorithm terminates in at most $m^2$ rounds.
        
        Given a specific preference graph $G$, we can define the skewness of the graph $W(G)$ to equal the skewness of the maximally skewed set. Therefore, by defining a sequence $W(G_t)=\max_{S \in \mathcal{M}, S \neq \emptyset} f_t(S)$, where $G_t$ is the preference graph obtained at the $t^{\mathrm{th}}$ iteration of Algorithm \ref{alg:alg1}, we show the convergence of the algorithm in finite rounds by proving that $W(G_t)$ strictly decreases with the decrease at least some positive constant. Thus, $W(\cdot)$ is a potential function that will be shown to strictly decrease in every iteration of the algorithm in the proof of Lemma~\ref{lem:conv}.
	  	\begin{lem} \label{lem:conv}
        For any unit demand matching market with $m>1$ the sequence $\{W(G_t)\}_{t\geq 0}$ of the skewness value of the maximally skewed set in each round of Algorithm \ref{alg:alg1} is strictly decreasing with minimum decrement $\tfrac{1}{m^2-m}$.
        \end{lem}     
	Given the minimum decrement in Lemma \ref{lem:conv}, it is straightforward that the preference graphs converge to the bipartite graph with combinatorial structure of MCP in time upper bounded by $m^3$ because $W(G)<m$. However, as there are only $m^2$ positive distinct feasible values of $W(G)$ \footnote{Since there are only $m$ possible values of $|S|-|N(S)|$ and $m$ possible values of $\frac{1}{|S|}$.}, we are ensured convergence in time at most $m^2$.

      \subsection{Complexity of the Algorithm}       
		Based on the results thus far determining the complexity of our algorithm depends only on the run-time of finding the maximally skewed set. We now discuss two approaches for this. 
		\subsubsection{Algorithm design in search of the maximally skewed set}		
		Given the uniqueness, we can always perform a brute-force search to get the maximally skewed set. Since there are $2^m-1$ non-empty subsets of $\mathcal{M}$, the complexity is $O(2^m)$, which doesn't meet our goal. We will exploit the combinatorial structure of the preference graph to scale down the complexity of finding the maximally skewed set. For this we design a graph coloring algorithm to color the preference graphs in Algorithm \ref{alg:alg3}.

 \begin{dfn}
    	A colored preference graph $G(\mathcal{M},\mathcal{B}, E)$ is an undirected graph that colors each vertex (goods and buyers) in three colors either red, green, or blue; and each edge is colored red or blue. Denote $X_c, X=\{\mathcal{M},\mathcal{B}\}, c=\{r,g,b\}$ to be the set of goods/buyers colored red, green or blue. $E^{gb}_{rb}$ denotes edges connecting goods in $\mathcal{M}_g \cup \mathcal{M}_b$ and buyers in $\mathcal{B}_r \cup \mathcal{B}_b$.
    \end{dfn}
    
    In any colored preference graphs, we want red edges to represent edges connecting matched pairs of good-buyer in a maximum matching, and blue edges to represent the rest of the edges. Hence, each vertex has at most one red edge. Additionally, we want the set of red goods $\mathcal{M}_r$ to represent the set of goods not in the maximally skewed set, the set of blue goods $\mathcal{M}_b$ to be goods in the maximally skewed set but ones that do not have matched pairs to buyers in this maximum matching (because of the nature of constricted good set),  and the set of green goods $\mathcal{M}_g$ are the rest of the goods. On the buyer side, the buyers that are neighbors of the maximally skewed set should be colored green, the buyers that are not the neighbors of the maximally skewed set but have a matched good should be colored red, and the rest of the buyers should be colored blue. Given the object we seek, we now present an algorithm to color vertices/edges properly in strongly polynomial-time complexity. The steps will include an initial coloring and followed by an update of the preference graph.
    
    Before detailing the initial coloring, we define various depth-first search and breadth-first search procedures relevant to the algorithm.
    \begin{dfn}
    A rb-DFS in $G(\mathcal{M},\mathcal{B},E)$ is a depth-first search (DFS) only using red edges from $\mathcal{M}$ to $\mathcal{B}$ and only using blue edges from $\mathcal{B}$ to $\mathcal{M}$. Similarly, a br-BFS in $G(\mathcal{M},\mathcal{B},E)$ is a breadth-first search (BFS) only using blue edges from $\mathcal{M}$ to $\mathcal{B}$ and only using red edges from $\mathcal{B}$ to $\mathcal{M}$. The set of nodes obtained at the end of the procedure will be called the reachable set $Rch(\cdot)$.
    \end{dfn}

\paragraph{Initial Coloring} First, we find a maximum matching using the Hopcroft-Karp algorithm and color edges linked matched pairs red, and other edges blue. After that, we start from the set of good without matched buyer in this maximum matching, color them blue, run the br-BFS algorithm starting from the set of blue goods. When the br-BFS algorithm terminates, color the set of reachable goods $Rch(\mathcal{M}_b)\cap \mathcal{M}$ with matched buyers green, color the rest set of goods red. Then, color the matched buyers of red goods red, color the buyers in the $Rch(\mathcal{M}_b)\cap \mathcal{B}$ green (they are the neighbors of the most skewed set), and color the rest of buyer blue. Finally, the following lemma states that we get the maximally skewed set from the initial coloring.
	
	\begin{lem} \label{lem:colorMSS}
	After the initial coloring, $\{\mathcal{M}_g \cup \mathcal{M}_b\}$ is the maximally skewed set.
	\end{lem}
	
	Given that the Hopcroft-Karp algorithm has complexity $O(m^{2.5})$ and the br-BFS has complexity upper-bounded by $O(m^2)$, we learn the initial coloring has the complexity $O(m^{2.5})$. Note that the initial coloring does not rely on the initial price, our first algorithm, say \textbf{\textit{initial coloring based decreasing price auction}}, will use this in every iteration to get the maximally skewed set. This algorithm has complexity $O(m^2\times m^{2.5})=O(m^{4.5})$, which is already strongly polynomial.
	\begin{lem} \label{lem:alg2}
	Given a bipartite graph with no perfect matching, initial coloring returns the maximally skewed set of this bipartite graph in a strongly polynomial run time of $O(m^{2.5})$.
	\end{lem}
	
\paragraph{Update the preference graph}	
	We further scale down the complexity of the \textbf{\textit{initial coloring based decreasing price auction}} algorithm by exploiting and updating the colored preference graph colored in previous round without completely coloring the whole graph. This is detailed in Algorithm \ref{alg:alg3}. Since the procedure in Algorithm \ref{alg:alg3} is elaborate, we will highlight some facts of perfect matching and give the sketch of how we use them in the Algorithm \ref{alg:alg3}.
	
	First, we know that if there is a perfect matching, no vertex should be colored blue otherwise we fail to get a maximum matching. Furthermore, the following Lemma \ref{lem:color} states that we will never need to change a vertex from red or green to blue. 
	\begin{lem} \label{lem:color}
	If we have to change the color of a vertex from red or green to blue based on interpretation given to the colors, one of the following is true: a) The coloring of previous preference graph is incorrect; b) The price reduction is not optimal; c) One of the updates from the previous preference graph to current graph is wrong.
	\end{lem}
	
Given the property of a proper coloring stated in Lemma \ref{lem:color}, we know that the set of blue buyers in round $t$, denoted as $\mathcal{B}_b^t$, is a decreasing set. Therefore, the key idea of designing Algorithm \ref{alg:alg3} is to restrict operations irrelevant to reducing the set of blue buyers to some constant number of $O(m^2)$ operations, and to allow the complexity of operations reducing the set of blue buyers to be upper-bounded by $O(m^3)$. To achieve this, we observe that the set of buyers at round $t$ that are willing to get goods in previous maximally skewed set, denoted by $A_t$, only contains red and blue buyers, i.e., $A_t\subseteq \mathcal{B}_r^t \cup\mathcal{B}_b^t$. Then, we know that if $A_t \subseteq \mathcal{B}_r$ and we cannot reach any blue buyers from $A_t$ without passing green or blue goods, then the current matching is maximized, and all we need to do for updating colors is to run rb-BFS from the set of $\mathcal{M}_b$ with complexity $O(m^2)$ as we did in initial coloring. In other cases, we need to find the new maximum matching with number of matched pairs increased by at least 1, which can be achieved by at least $O(m^{2.5})$ using the Hopcroft-Karp algorithm, and at least one blue buyer will be recolored. Since the maximum number of matched pairs is $m$, the complexity of the whole algorithm attaining the maximum MCP will be upper-bounded by the \{(complexity of updating process not recoloring blue buyers)$+$(complexity of computing price reduction)\}$\times$ (convergence rate of the preference graph)$+$(complexity of updating process increasing maximum matched pairs)$\times$ (maximum number of blue buyers)$=O((m^2+m^2)\times m^2+m^{2.5}\times m)=O(m^4)$.


    \begin{thm} \label{thm1}
    	The skewed-aided descending price algorithm has a strongly polynomial run time of $O(m^4)$ by using Algorithm \ref{alg:alg3} to search the maximally skewed set.
    \end{thm}

        
       

\section{Conclusions and Future Work}  \label{sec7}
    In this paper, we proposed a descending price algorithm in search of sets of the maximum MCPs by exploiting the combinatorial structure of bipartite graphs in matching markets. The algorithm terminates in at most $m^2$ rounds for any non-negative valuation matrix with runtime $O(m^4)$. There are three main avenues for future work. First, we would like to determine whether one can reduce the complexity further to $O(m^3)$ mirroring the Hungarian algorithm. Second, as incentive compatibility does not hold with the maximum MCP, we would like to determine the equilibrium bidding strategy in a Bayesian Nash equilibrium given the proposed mechanism. This will be necessary for expected revenue computation, and for a comparisons with the VCG mechanism, GSP and laddered auction proposed in \cite{GSP,ladder}, and also other mechanisms such as the GFP auction. 
	Finally, many real-world applications of matching markets outside of the online advertising setting are not unit-demand~\cite{oviedo2005theory},
	and obtaining a combinatorial version of the descending price auction returning the maximum MCP is a challenging open problem for future work.

\bibliographystyle{IEEEtran}
\bibliography{refact}
\newpage
\appendix
\section{Appendix-Details of Algorithms}
  \begin{algorithm}[H]
\caption{Algorithm in search of the maximally skewed set by coloring preference graph}
\begin{algorithmic}[1]
\Require
            A colored preference graph, a set of buyers $A$ would be added to the neighbor of the previous maximally skewed set
\Ensure
            An updated colored preference graph
\If {Input preference graph is not colored}
	\State Run the algorithm for initial coloring
\Else
	\State Update the colored preference graph by removing all edges connecting red goods and green buyers and adding corresponding blue edges connecting goods and buyers in $A$
    \While {$\{A\cap\mathcal{B}_b\} \neq \emptyset$}
		\State Pick an element $a$ in $\{A\cap\mathcal{B}_b\}$
		\If {$|N(a) \cap \mathcal{M}_b|>1$}
			\State Pick arbitrary $x \in \{N(a)\cap \mathcal{M}_b\}$, color $(a,x) \in E$ red and color $a,x$ green.
        \ElsIf {$|\{N(a) \cap \mathcal{M}_b\}|=1$}					\State Let $x$ be the unique good in $\{N(a)\cap \mathcal{M}_b\}$, color $(a,x) \in E$ red and color $a,x$ green.
            \If{$\{N(a)\cap \mathcal{M}_g\}\neq \emptyset$}
            	\State Run the rb-DFS starting from $x$ in $G(\mathcal{M}_g \cup \mathcal{M}_b,\mathcal{B}_g, E^{gb}_{gb})$ to get a reachable set $Rch(x)$, then color vertice in $Rch(S)$ red if $Rch(x)\cap \mathcal{M}_b = \{x\}$
            \Else               	
                \State Run the br-DFS starting from $a$ in $G(\mathcal{M}_g \cup \mathcal{M}_b,\mathcal{B}_g, E^{gb}_{gb})$ to get a reachable set $Rch(a)$, then color vertice in $Rch(S)$ red if $Rch(a)\cap \mathcal{M}_b = \{x\}$
           \EndIf
        \ElsIf {$\{N(a) \cap \mathcal{M}_g\} \neq \emptyset$}
           \State Run the rb-DFS starting from $a$ in $G(\mathcal{M}_g\cup \mathcal{M}_b,\mathcal{B}_g, E^{gb}_{gb})$ till find the first $x \in \mathcal{M}_b$
           \State Color $a,x$ green and switch the color of every edge used in a path from $a$ to $x$.
           \State Start a br-BFS from $\mathcal{M}_b$ in $G(\mathcal{M}_g\cup \mathcal{M}_b,\mathcal{B}_g, E^{gb}_{g})$ to get $Rch(\mathcal{M}_b)$.
           \State Color every vertex in $\{\mathcal{M}_g \cup \mathcal{B}_g\}\setminus Rch(\mathcal{M}_b)$ red
       \EndIf
       \State Remove $a$ from $A$       
	\EndWhile
   		\State Run the br-BFS starting from $\{\mathcal{M}_g\}$ in $G(\mathcal{M},\mathcal{B}, E)$ to get a reachable set $Rch(S^*)$
   		\If {$\{Rch(S^*) \cap \mathcal{B}_b\}=\emptyset$}
                \State Color all vertice in $Rch(S)\setminus \mathcal{M}_b$ green
        \Else
        	\While{$\{Rch(S^*) \cap \mathcal{B}_b\}\neq\emptyset$}
        		\State Pick $a \in Rch(S^*) \cap \mathcal{B}_b$ and run rb-DFS starting from $a$ to get $Rch(a)$
				\If {$\exists x\in Rch(a), x \in \mathcal{M}_b$}
				\State Color $a,x$ green; switch the color of every edge used in a path from $a$ to $x$
           		\State Start a br-BFS from $\mathcal{M}_b$ in $G(\mathcal{M}_g\cup \mathcal{M}_b,\mathcal{B}_g, E^{gb}_{g})$ to get $Rch(\mathcal{M}_b)$.
           		\State Color every vertex in $\{\mathcal{M}_g \cup \mathcal{B}_g\}\setminus Rch(\mathcal{M}_b)$ red
				\EndIf        		
        	\EndWhile   
        	\State Run the br-BFS starting from $\{\mathcal{M}_b\}$ in $G(\mathcal{M},\mathcal{B}_g, E)$ to get $Rch(\mathcal{M}_b)$
            \State Color all vertice in $\{\mathcal{M}_g\cup \mathcal{B}_g\}\setminus Rch(\mathcal{M}_b)$ red  
		\EndIf
\EndIf
            \end{algorithmic}
            \label{alg:alg3}
        \end{algorithm}
   \begin{algorithm}[h]
            \caption{Algorithm for initial coloring}
            \begin{algorithmic}[1]
            \Require
            A preference graph with no perfect matching
            \Ensure
            Colored preference graph, the maximally skewed set
            \State Find a maximum matching by Hopcroft-Karp algorithm
            \State Color every edges connecting a matched pair red and all other edges blue, Color every vertex in a matched pair red and all other vertice blue
            \State Starting from $\mathcal{M}_b$, run the rb-BFS to get a reachable set $Rch(\mathcal{M}_b)$
            \State Color $Rch(\mathcal{M}_b)\setminus \mathcal{M}_b$ green
            \end{algorithmic}
            \label{alg:alg2}
        \end{algorithm}
\section{Appendix-Discussion}
\subsection{Importance of Maximally Skewed Set} \label{sec6.1}
	To highlight the importance of choosing the maximally skewed set in our algorithm, we are starting to ask a natural question: if we run the algorithm twice but choose different constricted good sets at some iterations such that the preference graph produced in every round of these two executions are the same, will we get the same MCP vector?
		
        Unfortunately, the answer is no. The choice of the same initial price vector, the same price reduction rule, and the emergence of the same bipartite graph in every round are not enough to guarantee the same returned MCPs. A counterexample is provided in Fig. \ref{fig4}.
	
	\begin{figure}[h]
	\centering
        \includegraphics[width=0.4\textwidth]{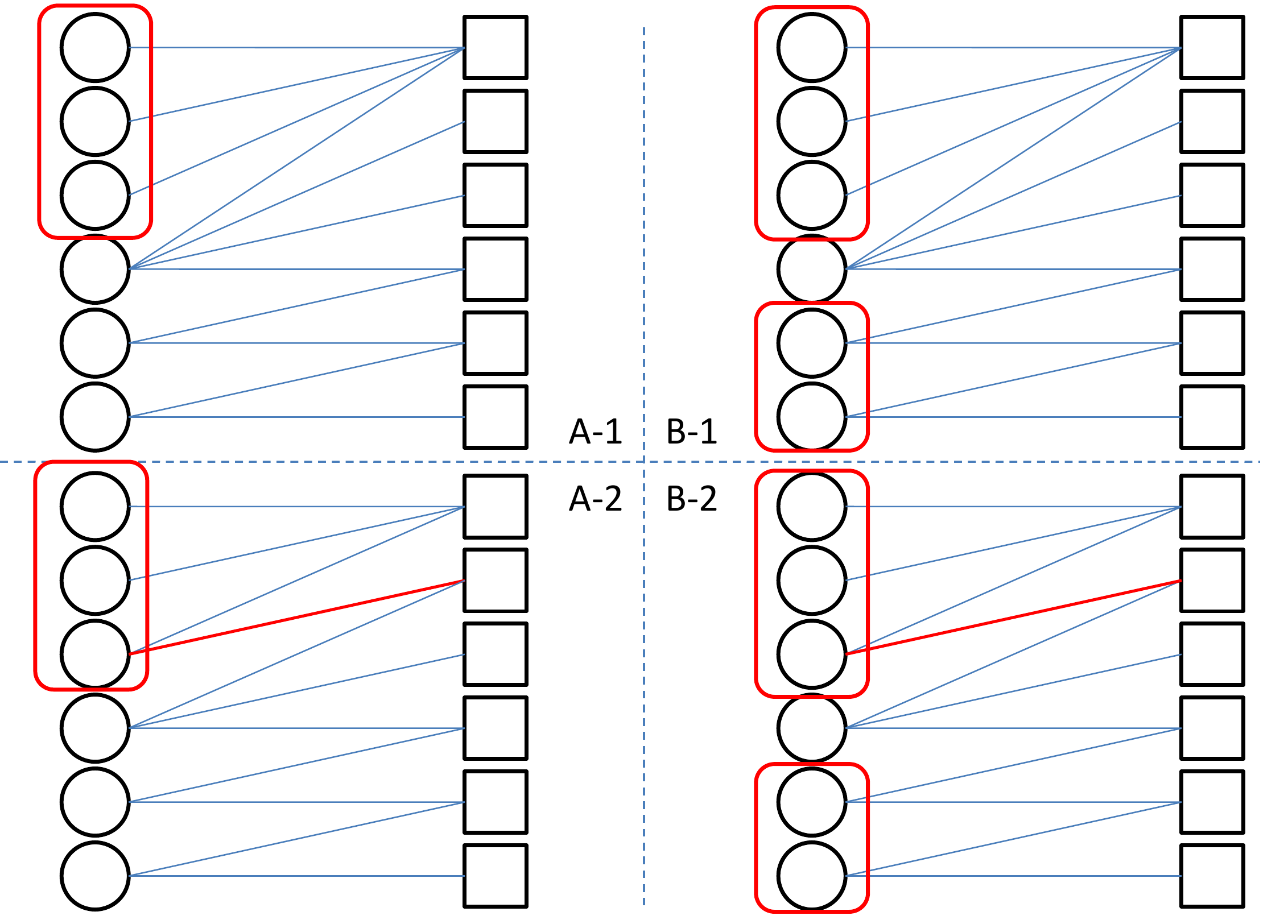}
        \caption{Counterexample of same bipartite graph but different set MCPs}
        \label{fig4}
    \end{figure}

    In Fig. \ref{fig4}, the bipartite graphs A-1, B-1 have the same preference graph. Though the chosen constricted good sets in A, B are different, they add the same buyer to the constricted graph. Therefore, the preference graphs in A-2, B-2 are still the same. However, the updated price vector of A-2, B-2 must be different. If A, B choose the same constricted goods set in every round, the returned sets of MCPs of A, B must be different. This example shows that just the bipartite graphs in every round cannot uniquely determine the MCP vector obtained at the termination of the algorithm.
    
    Then, we state a stronger claim in Lemma \ref{lem:multicon} that if an algorithm wants to get the maximum MCP, the law of choosing the constricted good set must pick the maximally skewed set at the round before termination.
    \begin{lem} \label{lem:multicon}
    Given a descending price algorithm with a specific law of choosing the constricted good set, and assuming this algorithm terminates at round $T_V$ when giving valuation matrix $V$m, if this law may not pick the most skewed set at round $T_V-1$, this algorithm not always return the maximum MCP.
    \end{lem}

	However, if a law of choosing the constricted good set pick the maximally skewed on in the last round before termination, the answer of whether it returns the maximum MCP will be case by case. Since Fig. \ref{fig4} provide an example of not achieving the maximum MCP, we now provide an example which still gets the maximum MCP in Fig.~\ref{fig5}.

	\begin{figure}[h]
	\centering
        \includegraphics[width=0.6\textwidth]{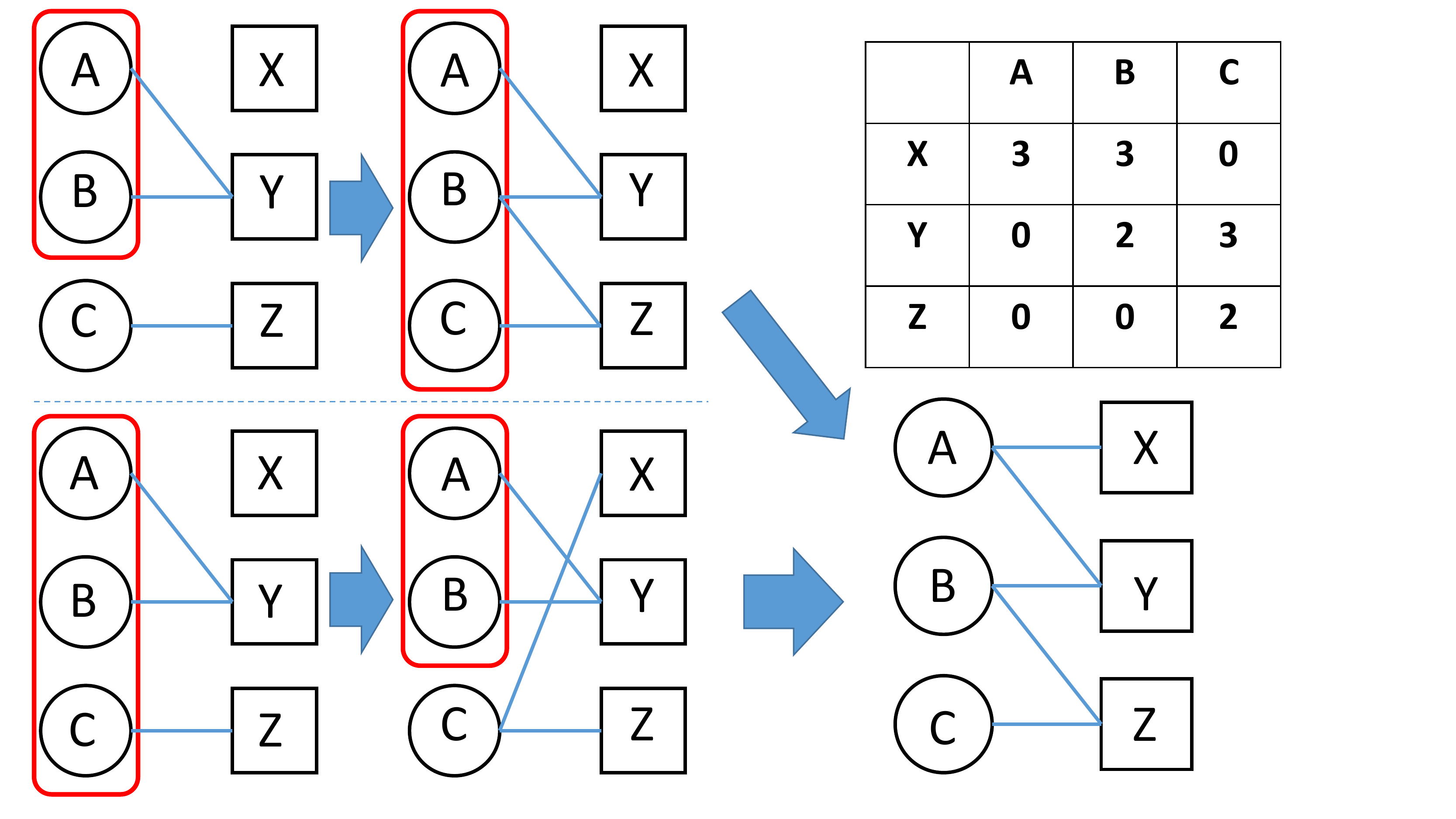}
        \caption{Choosing non-maximally skewed set in the middle but get maximum MCP}
        \label{fig5}
    \end{figure}
	In Fig.~\ref{fig5}, the algorithm in upper flow runs the algorithm we proposed and choose maximally skewed set in each round. We pick $\{A,B\}$ in the first iteration and then pick $\{A,B,C\}$ in the second iteration; and the algorithm terminates at $P=[1,1,2]$.
	 The algorithm in lower flow does not choose the maximally skewed set in the first round, but the algorithm is forced to choose the maximally skewed set in the second round because it is the only constricted good set. It picked $\{A,B,C\}$ and lower the price to $[2,2,2]$; and pick $\{A,B\}$ to get a MCP at $[1,1,2]$. 
	 As shown in the figure, these two algorithm both return the maximal MCP. Therefore, the space of sets we can choose which guaranteeing to get the maximum MCP is still an open problem.
\subsection{Interpretation of the Maximum MCP Using Buyers' Externalities,}
Since the minimum posted price derived from DGS algorithm corresponds to the personalized VCG price given by the Clarke pivot rule \cite{VCG_C}, it is not surprising that there is an analogous structure between the posted price and the personalized price also for the maximum MCP. The Clarke pivot rule determines the VCG payment of buyer $i$ using the externality that a buyer imposes on the others by her/his presence, i.e., payment of buyer $i=$ (social welfare\footnote{Social welfare is defined as the sum of buyers' surplus plus the sum of goods' prices, which is the same as the sum of each buyer's value on her matched good.} 
of others if buyer $i$ were absent) - (social welfare of others when buyer $i$ is present). Using the combinatorial characterization of the maximum MCP, we obtain an exact analogue of the Clarke pivot rule  when viewing the maximum MCP as the personalized price in Theorem~\ref{thm:externality}. In the theorem below we can use any perfect matching to determine the good matched to buyer $i$.

    \begin{thm} \label{thm:externality}
	Under the maximum MCP, the price that buyer $i$ pays is (social welfare of the current market adding a duplicate pair of buyer $i$ and its matched good) - (social welfare of the current market adding a duplicate buyer $i$). This price is the same as the decrease in social welfare that results by removing the good matched to buyer $i$.
    \end{thm}
\section{Appendix-Preliminary Analysis of Strategic Buyers and\\ Bayesian-Nash Equilibrium} \label{sec5}

As mentioned in the introduction, the proposed descending price algorithm is not incentive compatible so that we need to explore Bayesian Nash equilibrium (BNE) to predict buyers' strategic behaviors. From the messaging viewpoint we will assume that we have a direct mechanism where the buyers bid their valuation: a scalar in sponsored search markets, and a vector in the general matching market. Algorithm \ref{alg:alg1} is then applied as a black-box to the inputs to produce a price on each good and a perfect matching.  Although the strategic behavior of buyers is not the main subject of this work, with expected revenue being an important concern, we provide an instance achieving higher expected revenue than VCG mechanism in a BNE with asymmetric distributions of buyers' valuations. Furthermore, we analyze the BNE in two simple cases with symmetric buyers and the associated bidding strategy of the buyers.

\subsection{An Instance Achieving higher expected revenue than VCG}\label{sec:greatRev}

	Given the proposed algorithm, it is important to check if there exists a scenario that our algorithm can achieve a higher revenue than the VCG mechanism. It is well-known that the revenue equivalence theorem holds with an assumption that valuation of every player is drawn from a path-connected space in Chap. 9.6.3 of \cite{SSMbook}\footnote{Two other assumptions are (1) both mechanisms implement the same social welfare function; (2) for every player $i$, there exists a type that the expected payments are the same in both mechanisms}. 
Furthermore, there's a large body of literature that has discussed the failure of getting the VCG revenue under asymmetric distribution of buyers' valuations, e.g., see \cite{VCGfailure}. With this knowledge, we demonstrate a $3\times 3$ matching market where our mechanism achieves higher than VCG revenue, where the buyers have asymmetric distributions of their valuations.

    Consider three advertisers named Alice, Bob, and Carol, and three different types of ads called listing ads, sidebar ads, and pop-ups. The realized valuation is only known to the advertiser (equivalently buyer), but the distribution of an advertisers' valuation is known to other advertisers but not the auctioneer. In other words, the auctioneer can only calculate the price according to the bids submitted by the advertisers. The minimum increment of the submitted bids is $\epsilon$, which is a positive infinitesimal, and the valuation matrix of advertisers is displayed in Table \ref{tbl:1}.
    \begin{table}[h]
\centering
\caption{Valuation matrix of advertisers}
\label{tbl:1}
\begin{tabular}{|l|l|l|l|} \hline
      & Listing & Sidebar & Pop-ups   \\ \hline
Alice & w & 0 & 0   \\ \hline
Bob   & x & 0  & 0.5 \\ \hline
Carol & y & z & 2    \\ \hline
\end{tabular} \\
\begin{tabular}{|ll|} \hline
Probability density function of $w,x,y,z$ &\\ \hline
$f_w(w)=\begin{cases}
	\frac{2}{3} &~w\in [0,1) \\
    \frac{1}{3} &~w\in (2,3] \\
    0 &~\text{o/w}
	 \end{cases} $ &
$f_x(x)=\begin{cases}
	\frac{2}{3} &~x\in [1.5,2.5] \\
    \frac{1}{3} &~x\in (2.5,3.5] \\
    0 &~\text{o/w}
	 \end{cases}$      \\
$f_y(y)=\begin{cases}
	2 &~y\in [3.5,4] \\
    0 &~\text{o/w}
	 \end{cases}$
&$f_z(z)=\begin{cases}
	1 &~z\in [3,4] \\
    0 &~\text{o/w}
	 \end{cases}   $  \\ \hline
\end{tabular}
\end{table}

Now, we give an asymmetric BNE\footnote{There need not be a unique equilibrium.} in this matching market and provide a detailed verification in the Appendix \ref{asymBNE}. First, consider Alice always bids 0 on sidebar ads and pop-ups, and bids $\max \{1-\epsilon, \frac{w}{2}\}$ on listing ads. The best response of Bob is to bid 0 on both sidebar ads and pop-ups, and to bid $\max \{1, \frac{x-0.5}{2}\}$ on listing ads for any realized $x$. Now, given Bob's bidding function as above, one of Alice's best responses is to follow her original bidding function. Last, consider the bidding function of Alice and Bob mentioned above, a best response of Carol is to bid 0 on listing ads and pop-ups, and to bid $\epsilon$ on sidebar ads regardless of the outcomes of $y$ and $z$.   
\footnote{The $\epsilon$ is designed to avoid complex tie-breaking rules.}, even if $y>z$. This is because Carol will never win listing ads for any bids less than $1$ as the probability $\text{Pr}\{y-z \geq 1\}=0$. Now, given Carol's bidding strategy, Alice and Bob will not change their bidding functions. Therefore, the strategy $\{\beta_{Alice}(w,0,0), \beta_{Bob}(x,0,0.5),\beta_{Carol}(y,z,2)\}=\{(\max \{1-\epsilon, \frac{w}{2}\},0,0), (\max \{1, \frac{2x-1}{4}\},0,0),(0,\epsilon,0)\}$ can be verified as an asymmetric BNE.

	With the asymmetric BNE in hand, we want to calculate the expected revenue of the auctioneers and compare it with the expected revenue of the VCG mechanism. Since Carol always wins the sidebar ads and both Alice and Bob bid 0 on that, Carol will pay $\epsilon$. Additionally, in the asymmetric BNE, Alice and Bob compete on the listing ads and both bid 0 on sidebar ads and pop-ups, resulting in the payment of listing ads to be the same as the payment in the first price auction. Next, let's calculate the expected revenue of auctioneers, which is given by
    \begin{eqnarray}
    &&\epsilon+\frac{4}{9}\int_0^1\int_{1.5}^{2.5} 1 dxdw+\frac{2}{9}\int_2^3\int_{1.5}^{2.5} \frac{w}{2} dxdw+\frac{2}{9}\int_1^2\int_{2.5}^{3.5} \frac{2x-1}{4} dxdw \nonumber \\&+&\frac{1}{9}\int_2^3\int_{2.5}^{w+0.5} \frac{w}{2} dxdw+\frac{1}{9}\int_{2.5}^{3.5}\int_2^{x-0.5} \frac{2x-1}{4} dwdx =\dfrac{31}{27}+\epsilon
    \end{eqnarray}
The last step is to calculate the expected revenue of the VCG mechanism, which is given by
    \begin{eqnarray}
    &&\frac{2}{3}\int_0^1\int_{1.5}^{2.5} w dxdw+\frac{2}{9}\int_2^3\int_{1.5}^{2.5} (x-0.5) dxdw+\frac{1}{9}\int_{2.5}^{3.5}\int_2^{x-0.5} w dwdx \nonumber \\&+&\frac{1}{9}\int_2^3\int_{2.5}^{w+0.5} (x-0.5) dxdw =\dfrac{1}{3}+\dfrac{2}{9}\times \dfrac{3}{2}+\dfrac{1}{9}\times \dfrac{7}{6}++\dfrac{1}{9}\times \dfrac{7}{6}=\dfrac{25}{27}
    \end{eqnarray}
Even if we set $\epsilon$ to 0, it is obvious that the expected revenue derived under our descending price auction algorithm is strictly greater than the expected revenue of the VCG mechanism. This shows that in some instances the proposed descending price algorithm is preferred to the (DGS) ascending price algorithm,  even taking the strategic behavior of buyers into account.

\subsection{Symmetric Bayesian Nash Equilibrium in $2\times 2$ Sponsored Search Market}
	
    Now we focus on the analysis of symmetric distributions of buyers. Next we will detail the analysis of two cases for a market with two goods and two buyers, one in the following paragraphs and another one in the next subsection.

    Since the primary application of our algorithm is in online advertising auctions, it is useful for us to analyze the strategic behavior under the conventional assumptions made in the sponsored search market setting. The sponsored search market assumes every buyer's (advertiser's) value on goods (web slots) can be determined by a product of the buyer's private weight and a common click-through rate. Now, we consider a $2 \times 2$ case in sponsored search market with settings detailed below.

   There are two web slots with click-through rates $c_1$ and $c_2$, with  $c_1\geq c_2$, and two advertisers with private weights $w_1, w_2$. We assume that $w_j$ is an \emph{i.i.d.} non-negative random variable with PDF $f_{w_j}(\cdot)$, and the private value for getting web slots $j$ is $w_ic_j$. Each advertiser knows his/her true weight but only knows the distribution of another advertiser's value, and both know that it is a sponsored search market. Under our descending price auction algorithm, they each have to effectively place a one-dimensional bid, denoted by $b_i$, according to their bidding function $\beta_j(w_j,f_{w_{-j}}(\cdot))$, where ${-j}$ denotes the other advertiser(s). To simplify the analysis, we assume weights $w_1, w_2$ are  uniformly distributed in $[0,1]$. Extensions to other symmetric distributions, asymmetric distribution/knowledge space and more web slots are all for future work.

    \subsubsection{Payment of advertisers}
        If $c_2=0$, the descending price algorithm terminates at the initial point. At this point, the advertiser $j$ wins the first web slots pays $c_1w_j$ and another advertiser pays $0$. This is, equivalent, to the single good case, which has been carefully studied before. When $c_2 \neq 0$, the descending price algorithm makes the advertiser indifferent between two slots if he/she bid truthfully. Therefore, if advertiser $i$ wins the first slot, the payment will be $(c_1-c_2)b_i+c_2b_{-i}$, if he/she gets the second slot, the payment will be $c_2b_i$.
    \subsubsection{Analysis of strategic behavior}
    	Before deriving the bidding function, we show that the bidding function is a monotonic increasing function of the weight. Since the outcome of this auction is exactly the same as what the VCG mechanism provides, this monotonicity leads us to expect revenue equivalence to hold in this case. Exploring the generality of this result is for future work.
        \begin{lem} \label{lem:mono}
        The bidding function is monotonically increasing in the $2\times 2$ sponsored search market.
        \end{lem}
        \begin{proof}
        See Appendix \ref{pf:lem:mono}.
        \end{proof}
        This monotonicity property of bidding functions can be generalized to all sponsored search markets when every advertiser's private weight follows the same uniform distributions. 
         \begin{cor} \label{cor:mono}
 In sponsored search markets with a symmetric uniform distribution of every advertiser's weight, the bidding function is monotonic and the allocation is always efficient.
 \end{cor}
        \begin{proof}
        See Appendix \ref{pf:cor:mono}.
        \end{proof}
        Without loss of generality, we now derive the bidding function of advertiser 1. Since the optimal bidding function is monotonic in the private value, the surplus function of advertiser 1 can be written using an integral form.  Advertiser 1 wants to maximize the following surplus function\footnote{In this two advertisers case, for advertiser 1, $f_{w_{-1}}(\cdot)=f(w_2)(\cdot)$ and $\beta_{-1}(\cdot)=\beta_{2}(\cdot)$.}:
        \begin{eqnarray}
         \int_0^{\beta_2^{-1}(b_1)}(c_1w_1-[(c_1-c_2)b_1 +c_2\beta_2(x)])f_{w_2}(x)dx +c_2(w_1-b_1)[1-\int_0^{\beta_2^{-1}(b_1)}f_{w_2}(x)dx] 
        \end{eqnarray}
        With detailed analysis presented in Appendix \ref{apdx:analysis}, the bidding function $\beta_1(w_1)$ is
    \begin{eqnarray}
    \begin{cases}
    	e^{-w_1}+w_1-1 \qquad \qquad \qquad \qquad \qquad \qquad ~~~c_1=2c_2 \\
        w_1-(2-w_1)\ln(1-0.5w_1) \qquad \qquad \qquad ~~~~c_1=1.5c_2 \\
    	\dfrac{c_2}{2c_1-3c_2}\Big[(c_1-c_2)w_1+\big(\dfrac{c_2}{c_2+(c_1-2c_2)w_1}\big)^{\tfrac{c_1-c_2}{c_1-2c_2}}-1\Big] \\ \nonumber 	
        \qquad \qquad \qquad \qquad \qquad \qquad \qquad \qquad  \qquad \qquad \text{otherwise} \nonumber
    	\end{cases}
    \end{eqnarray}
Without loss of generality, we can let $c_1=1$ and $c_2 \in [0,1]$, Figure \ref{fig5} displays the optimal bidding bid corresponds to the buyer's private weight and the ratio of click-through rates.
\begin{figure}[h]
	\centering
        \includegraphics[width=0.6\textwidth]{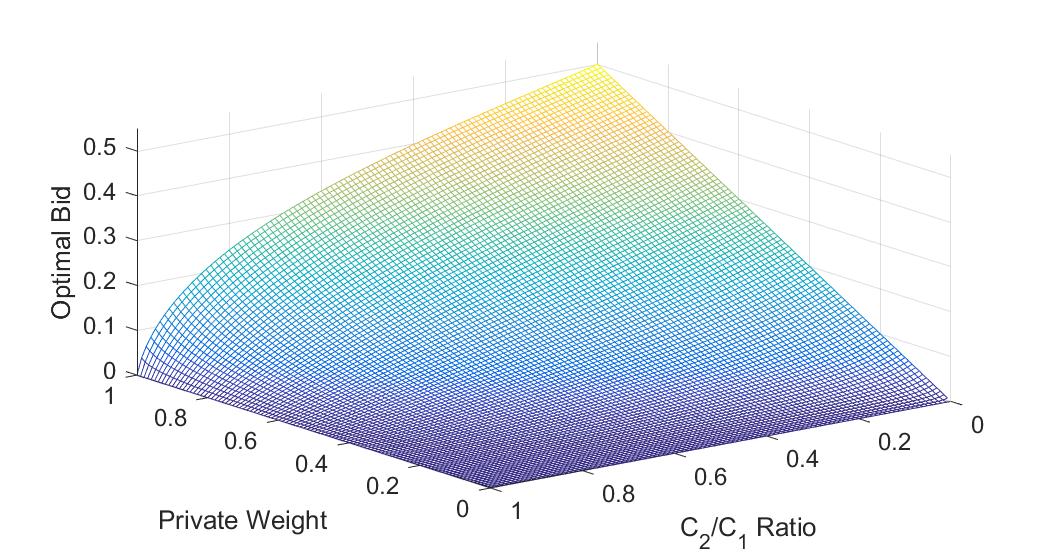}
        \caption{Optimal bidding function of the symmetric BNE in $2 \times 2$ sponsored search markets}
        \label{fig5}
    \end{figure}

    Then, the expected revenue of the search engine is
 \begin{eqnarray}
         &&2\bigg\{\int_0^1\int_0^{w_1}([(c_1-c_2) \beta_1(w_1)+c_2 \beta_1(w_2)])dw_2dw_1 +\int_0^1\int_0^{w_2}c_2\beta_1(w_2)dw_1dw_2 \bigg\} \nonumber\\
         &=& 2\bigg\{(c_1-c_2)\int\limits_0^1 w_1\beta_1(w_1)dw_1+2c_2\int\limits_0^1\int\limits_0^{w_1}\beta_1(x)dxdw_1 \bigg\} =\begin{cases}
    	\frac{c_2}{3}&~c_1=2c_2 \\
        \frac{c_2}{6}&~c_1=1.5c_2 \\
    	\frac{c_1-c_2}{3} 	&~\text{otherwise} \nonumber
    	\end{cases}
        \end{eqnarray}
        It is well known that in the VCG mechanism, the expected VCG revenue is $\dfrac{c_1-c_2}{3}$. Therefore, in this $2\times 2$ case, the proposed algorithm has the same expected revenue as the VCG mechanism after taking the strategic behavior into consideration.

\subsection{Symmetric Bayesian Nash Equilibrium in $2\times 2$ General Unit-demand Matching Markets}
Theoretically, the same method we used to studied the $2\times 2$ case of a sponsored search market can be used to solve the $2\times 2$ case for a general matching market when valuations are assumed to be \emph{i.i.d.}. However, that method requires us to consider two variables simultaneously, and we cannot avoid solving the resulting partial differential equations. Hence, before directly solving the general $2\times 2$ cases, we present the following lemmas to simplify the analysis afterward.

\begin{lem} \label{lem_2g1}
In a general $2\times 2$ matching market, any strategy placing non-zero bids on both goods is a weakly-dominated strategy of a rational buyer with non-zero valuation on goods.
\end{lem}
\begin{proof}
Please see Appendix \ref{sec:lem_2g1}.
\end{proof}

\begin{lem}  \label{lem_2g2}
If $v_{i1}>v_{i2}$, any bidding strategy of buyer $i$ with $b_{i1}<b_{i2}$ is a weakly dominated strategy.
\end{lem}
\begin{proof}
Please see Appendix \ref{sec:lem_2g2}.
\end{proof}

With the above two lemmas, the following corollary is straightforward.
\begin{cor}
Consider a rational buyer $i$, any bidding strategy putting non-zero bid on the good with private value lower than another good is a weakly-dominated strategy.
\end{cor}

Now, we can get rid of all weakly-dominated strategies to analyze the buyer's strategic behavior to find a symmetric Bayesian Nash Equilibrium. It is worth to note that weakly-dominated strategies can still be rationalizable strategies, our approach (analyzing Bayesian Nash equilibrium(BNE) without taking weakly-dominated strategies into consideration) may miss some BNEs. Our current goal is not to find all BNEs but just one. Therefore, we can continue to find a BNE with restricted strategy space.

First,  assume  $v_{i1}\geq v_{i2}$ without loss of generality, a rational buyer $i$ will bid 0 on the good 2, and bid no more than $v_{11}-v_{12}$ on good 1. Then, we can  compute the equilibrium in symmetric bidding strategies. Assume the symmetric bidding function is $\beta(\cdot,\cdot)$ for the good with higher value and 0 for the other one, buyer $i$ bids $b$ on good 1 and the CDF of buyer $-i$'s valuation on good $j$ is $F_{-ij}(\cdot)$, the objective function is
 \begin{eqnarray}
 \max_b  E_{(v_{-i1},v_{-i2})}  [u_i(b,0,\beta\mathbf{1}_{\{v_{-i1}>v_{-i2}\}},\beta\mathbf{1}_{\{v_{-i1}\leq v_{-i2}\}})],
 \end{eqnarray}
where $\beta$ denotes $\beta(v_{-i1},v_{-i2})$.

Now, we need Lemma \ref{pro1} to simplify our analysis.
\begin{lem} \label{pro1}
Denote $v_{ih}$ to be the buyer $i$'s value of the good that has higher value and $v_{il}$ be the buyer $i$'s value of good that has lower value, an optimal bidding function of the higher good $\beta(v_{ih},v_{il})$ is monotonic in $v_{ih}-v_{il}$.
\end{lem}
 \begin{proof}
 Please see Appendix \ref{sec:pro1}.
 \end{proof}
With Lemma \ref{pro1}, we can revise and simplify the bidding function $\beta(v_{ih},v_{il})$ to be $\beta_{r}(v_{ih}-v_{il})$. Define the joint CDF of $v_{-i1}-v_{-i2}$ to be $F_{-i}(\cdot)$. Then the expected surplus of buyer $i$ for a specific bid $b$ of good 1 is:
\begin{eqnarray}
    (v_{i1}-b)F_{-i}(\beta_r^{-1}(b))+v_{i2}(1-F_{-i}(\beta_r^{-1}(b)))
    =v_{i2}+(v_{i1}-v_{i2}-b_{i1})F_{-i}(\beta_r^{-1}(b))
\end{eqnarray}
Therefore, determining the objective function is equivalent to solving the following optimization problem:
\begin{eqnarray}
\max_b ~(v_{i1}-v_{i2}-b)F_{-i}(\beta_r^{-1}(b)).
\end{eqnarray}
To simplify our problem, let's assume the distributions of private value on goods of two buyers are i.i.d. uniformly distributed on $[0,1]$. Then assume that the bidding function is differentiable and denote $x=v_{i1}-v_{i2}$. Using the same technique we used to solve the symmetric BNE in sponsored search market. With the analysis in Appendix \ref{apdx:analysis2}, the bidding function $\beta_r(x)$ is
    \begin{eqnarray}
    \beta_r(x)= \dfrac{1}{2-(1-x)^2}\int_0^{x} 2\tau(1-\tau)d\tau
     =\dfrac{x^2(1-\tfrac{2x}{3})}{2-(1-x)^2}
    \end{eqnarray}
    Then, the expected revenue is
    \begin{eqnarray}
     4\int_0^1 \beta_r(x)(1-x)(\frac{1}{2}+x-\frac{x^2}{2}) dx
    =2\int_0^1 x^2(1-\tfrac{2x}{3})(1-x) dx
    = 2(\dfrac{1}{3} - \dfrac{5}{12} + \dfrac{2}{15})=  \dfrac{1}{10}
    \end{eqnarray}
    Finally, computing the expected revenue with VCG price yields:
    \begin{eqnarray}
     4\int_0^1 (1-x)\int_0^x y(1-y) dydx
    = 4\int_0^1 (1-x)(\dfrac{x^2}{2}-\dfrac{x^3}{3}) dx
    = 4 ( \dfrac{1}{6}- \dfrac{5}{24}+\dfrac{1}{15}) = \dfrac{1}{10}
    \end{eqnarray}
    It is shown here that the revenue equivalence theorem continues to hold in this case.
\section{Appendix-Proofs}        
\subsection{Proof of Lemma \ref{lemA}}\label{pf:lem1}
	First, we will show that  given a  constricted good set $S$, reducing price of every good in $S$ with the amount specified in the statement will increase the size of its neighbor, i.e., $N(S)$.
	
	Given a constricted good set $S$  under a specific price vector $\mathbf{P}$. Consider another price vector $\mathbf{P'}$, 
    $P'_j=\begin{cases}
    		P_j ,  & j \notin S \\
            P_j-c, &  j \in S
    \end{cases}
    $
, where $c:=\min_{i \in \mathcal{B} \setminus N(S),l\in S} \{\max_{k \in \mathcal{M} \setminus S}(v_{i,k}-P_k)- (v_{i,l}-P_l) \}$.

	For some $j \in S$, there must exists an $i\in \mathcal{B} \setminus N(S)$ satisfying $v_{i,j}-P_j'= \max_{k \in \mathcal{M} \setminus S} v_{i,k}-P_k$. Now, by an abuse of notation to denote $N'(S)$ as the neighbor of $S$ under $P'$, $i \in N'(S)$.
    For those $l \in N(S)$, $\max_{j \in S}v_{l,j}-P'_j>\max_{j \in S}v_{l,j}-P_j\geq \max_{j \in \mathcal{M}\setminus S}v_{l,j}-P_j$ implies $l \in N'(S)$. Therefore, $|N(S)|<|N'(S)|$.

Then, we need to prove that $c$ is the minimum decrement.
Consider another price vector Consider another price vector $\mathbf{P''}$, 
    $P'_j=\begin{cases}
    		P_j ,  & j \notin S \\
            P_j-d, &  j \in S
    \end{cases}
    $, where $d<c$. It is straightforward that for all $i \in \mathcal{B} \setminus N(S)$, $v_{i,j}-P''_j< v_{i,j}-P_j+c \le \max_k (v_{i,k}-P_k)$. Hence, no buyers will be added to the $N(S)$. Now, it is clear that $\min_{i \in \mathcal{B} \setminus N(S),l\in S} \{\max_{k \in \mathcal{M} \setminus S}(v_{i,k}-P_k)- (v_{i,l}-P_l) \}$ is the minimum price reduction which guarantees to add at least a new buyer to the $N(S)$.

\subsection{Proof of Lemma \ref{lem2}}\label{pf:lem2}
 In order to prove the statement, we start from showing the most skewed set is always a constricted good set when there is no perfect matching. Then, we prove the uniqueness of the most skewed set by contradiction.
 
 	When there is no perfect matching, there must exists a constricted good set. By definition, the constricted good set $S$ has the property $|S|>|N(S)|$. Since $|S|$, $|N(S)|$ are integers, the skewness of a constricted good set
	\begin{equation} \label{eq1}
	f(S)=|S|-|N(S)|+\frac{1}{|S|} \geq 1+\frac{1}{|S|} >1.
	\end{equation}
	Then, for any non-constricted good set $S'$, $|S'| \leq |N(S')|$. The skewness of $S'$ is
	\begin{equation} \label{eq2}
	f(S')=|S'|-|N(S')|+\frac{1}{|S'|} \leq 0+\frac{1}{|S'|}  \leq 1.
	\end{equation}	
	With equation (\ref{eq1}), (\ref{eq2}), the skewness of a constricted good set is always greater than any non-constricted good sets. Therefore, if a preference graph exists a constricted good set, the most skewed set is always a constricted good set.
    
    Then, we start to prove the uniqueness of the most skewed set.  Suppose there exists two disjoint sets $S_1$, $S_2$ and both sets are the most skewed sets, i.e., $f(S_1)=f(S_2)= \max_{S \subset \mathcal{M}, S \neq \emptyset} f(S)$. Consider the skewness of the union of $S_1$ and $S_2$.
	\begin{eqnarray}
	&&f(S_1 \cup S_2) \\
	&=& |S_1 \cup S_2|- |N(S_1 \cup S_2)|+ \frac{1}{|S_1 \cup S_2|} \\
	&=& |S_1|+ |S_2| - |N(S_1) \cup N(S_2)|+ \frac{1}{|S_1 \cup S_2|} \\
	&\geq & |S_1|+ |S_2| - |N(S_1)|- |N(S_2)|+ \frac{1}{|S_1 \cup S_2|} \\
	&=& f(S_1)+ |S_2|- |N(S_2)| -\frac{1}{|S_1|}+\frac{1}{|S_1 \cup S_2|} \label{eq3}\\
	&\geq & f(S_1) +1 -\frac{1}{|S_1|}+\frac{1}{|S_1 \cup S_2|} > f(S_1) \label{eq4}
	\end{eqnarray}
	From (\ref{eq3}) to (\ref{eq4}) is true because $S_2$ is a constricted good set. (\ref{eq4}) contradicts our assumption that  $S_1$, $S_2$ are the most skewed set. Therefore, we know that if there exists multiple sets share the same highest skewness value, these sets are not disjoint.
    
    Now,  suppose there exists two sets $S_1$, $S_2$ satisfying that $S_1 \cup S_2 \neq \emptyset$ and both sets are the most skewed sets, the following two inequalities must hold:
    \begin{eqnarray}
    f(S_1)-f(S_1 \cup S_2) \geq 0 \\
    f(S_2)-f(S_1 \cap S_2) \geq 0 
    \end{eqnarray}
	Let's sum up the two inequalities and represent the formula in twelve terms.
	\begin{eqnarray}
	&&f(S_1)+f(S_2) - f(S_1 \cup S_2)-f(S_1 \cap S_2) \\
	&=&	|S_1|+|S_2|-|S_1 \cup S_2|-|S_1 \cap S_2| \nonumber \\
	&& +|N(S_1 \cup S_2)|+ |N(S_1 \cap S_2)|- |N(S_1)|-|N(S_2)| \nonumber \\
	&& +\dfrac{1}{|S_1|}+\dfrac{1}{|S_2|}-\dfrac{1}{|S_1 \cup S_2|}-\dfrac{1}{|S_1 \cap S_2|}
	\end{eqnarray}
	The first four terms $|S_1|+|S_2|-|S_1 \cup S_2|-|S_1 \cap S_2|=0$.
	Using the similar argument, $|N(S_1)|+|N(S_2)|=|N(S_1) \cap N(S_2)|+|N(S_1) \cup N(S_2)|$.
	\begin{eqnarray}
	|N(S_1 \cup S_2)|=|N(S_1) \cup N(S_2)| \\
	|N(S_1 \cap S_2)| \leq |N(S_1) \cap N(S_2)| \label{eq5}
	\end{eqnarray}
	Equation (\ref{eq5}) is true because there may exist some elements in $S_1 \setminus S_2$ and $S_2 \setminus S_1$ but have common neighbors. Thus, the second four terms are smaller than or equal to $0$.
	
	To check the last four terms, let $|S_1|=a$, $|S_2|=b$, and $|S_1 \cap S_2|=c$, where $c<\min\{a,b\}$ because $S_1,S_2$ are not disjoint. The last four terms are
	\begin{eqnarray}
	&&\dfrac{1}{|S_1|}+\dfrac{1}{|S_2|}-\dfrac{1}{|S_1 \cup S_2|}-\dfrac{1}{|S_1 \cap S_2|}  \\
	&=& \dfrac{1}{a}+\dfrac{1}{b}-\dfrac{1}{a+b-c}-\dfrac{1}{c} \\
	&=& \dfrac{a+b}{ab}-\dfrac{a+b}{(a+b-c)c} \\
	&=& \dfrac{a+b}{abc(a+b-c)}(ac+bc-c^2-ab) \\
    &=&-\dfrac{(a+b)(a-c)(b-c)}{abc(a+b-c)}<0
	\end{eqnarray}
	To conclude, the first four terms are $0$, the second four terms are smaller than or equal to $0$, and the last four terms are strictly negative make
	\begin{eqnarray}
	f(S_1)+f(S_2) - f(S_1 \cup S_2)-f(S_1 \cap S_2)<0.
	\end{eqnarray}
	Therefore, at least one of set $S_1 \cup S_2$, $S_1 \cap S_2$ has the skewness value greater than $f(S_1)=f(S_2)$, which leads to a contradiction that $S_1$ and $S_2$ are the most skewed set. Finally, we can claim that the most skewed set is unique when there is no perfect matching.
\subsection{Proof of Lemma \ref{lem:conv}}\label{pf:lem:conv}
First, we prove that the algorithm terminates in finite (at most $|\mathcal{M}|^3$) rounds by investigating the relationship of the most skewed sets in the consecutive rounds, $S^*_t$, $S^*_{t+1}$. 

	The relationship of $S^*_t$, $S^*_{t+1}$ has four cases.
	\begin{enumerate}
	\item $S^*_t=S^*_{t+1}$
	\item $S^*_t \subset S^*_{t+1}$
	\item $S^*_t \supset S^*_{t+1}$
	\item $S^*_t \nsubseteq S^*_{t+1}$ and $S^*_t \nsupseteq S^*_{t+1}$
	\end{enumerate}
    Recall that $W(G)$ is the skewness of the graph $G$ and $N_t(S)$  is the neighbor of $S$ based on the preference graph at round $k$. \\
	In case 1, $W(G_t)-W(G_{t+1})=f_t(S^*_t)-f_{t+1}(S^*_{t+1})=f_t(S^*_t)-f_{t+1}(S^*_t)\geq 1$. \\
	In case 2, define $S'=S^*_{t+1}\setminus S^*_t$,  and it is trivial that $S^*_t \subset S^*_{t+1}$ implies $1>\frac{1}{|S^*_t|}-\frac{1}{|S^*_{t+1}|}>0$.
 		\begin{eqnarray}
	&&W(G_t)-W(G_{t+1})=f_t(S^*_t)-f_{t+1}(S^*_{t+1})\\
    &=&|S^*_t|-|N_t(S^*_t)|-|S^*_{t+1}|+|N_{t+1}(S^*_{t+1})| +\dfrac{1}{|S^*_t|}-\dfrac{1}{|S^*_{t+1}|} \label{eq6}\\
	&\geq &|S^*_{t+1}|-|N_t(S^*_{t+1})|-|S^*_{t+1}|+|N_{t+1}(S^*_{t+1})|+\dfrac{1}{|S^*_t|}-\dfrac{1}{|S^*_{t+1}|} \label{eq7}\\
	&=&|N_{t+1}(S^*_t \cup S')|-|N_t(S^*_{t+1})|+\frac{|S^*_{t+1}|-|S^*_t|}{|S^*_t||S^*_{t+1}|} \nonumber \\
	&\geq &|N_{t+1}(S^*_t \cup S')|-|N_t(S^*_{t+1})| + \dfrac{1}{|\mathcal{M}|(|\mathcal{M}|-1)}  \label{eq8}\\
	&=&|N_{t+1}(S^*_t)|+|N_{t+1}(S')\cap N_{t+1}(S^*_t)^c| -|N_t(S^*_{t+1})|+ \dfrac{1}{|\mathcal{M}|^2-|\mathcal{M}|}  \label{eq9}\\
	&=&|N_t(S^*_t)|+|N_{t+1}(S_t)\setminus N_t(S^*_t)|- |N_t(S^*_{t+1})| +|N_{t+1}(S')\cap N_{t+1}(S^*_t)^c|+ \tfrac{1}{|\mathcal{M}|^2-|\mathcal{M}|}  \label{eq10}
	\end{eqnarray}
	
	Before going to further steps, we have to briefly explain the logic behind the above equations.
	
	(\ref{eq6}) to (\ref{eq7}) is true because $N_t(S^*_t) \subseteq N_t(S^*_{t+1})$.
	
	(\ref{eq8}) to (\ref{eq9}) is to expand $|N_{t+1}(S^*_t \cup S')|$ to $|N_{t+1}(S^*_t)|+|N_{t+1}(S')\cap N_{t+1}(S^*_t)^c|$.
	
	(\ref{eq9}) to (\ref{eq10}) is to expand $|N_{t+1}(S^*_t)|$ to $|N_t(S^*_t)|+|N_{t+1}(S_t)\setminus N_t(S^*_t)|$.
	
	Since $|N_{t+1}(S')\cap N_{t+1}(S^*_t)^c|=|N_{t+1}(S')\cap N_{k}(S^*_t)^c|-|N_{t+1}(S')\cap (N_{t+1}(S_t)\setminus N_t(S^*_t))|$, and
	$|N_{t+1}(S_t)\setminus N_t(S^*_t)|\geq |N_{t+1}(S')\cap (N_{t+1}(S_t)\setminus N_t(S^*_t))|$, we can further summarize the first four terms in (\ref{eq10}).
	\begin{eqnarray}
	&&|N_t(S^*_t)|+|N_{t+1}(S_t)\setminus N_t(S^*_t)|  +|N_{t+1}(S')\cap N_{t+1}(S^*_t)^c| - |N_t(S^*_{t+1})| \\
	&\geq & |N_t(S^*_t)|+|N_{t+1}(S')\cap N_{k}(S^*_t)^c| -|N_t(S^*_{t+1})|\\
	&=& |N_t(S^*_t)|+|N_{t+1}(S')\cap N_t(S^*_t)^c| -|N_t(S^*_t)|-|N_t(S')\cap N_t(S^*_t)^c| \\
	&=& |N_{t+1}(S')\cap N_t(S^*_t)^c|-|N_t(S')\cap N_t(S^*_t)^c| \label{eq32} \\
    &=& |N_t(S')\cap N_t(S^*_t)^c|-|N_t(S')\cap N_t(S^*_t)^c|=0 \label{eq33}
	\end{eqnarray}
	
	(\ref{eq32}) to (\ref{eq33}) is true because $S'$ is not in $S^*_t$, the neighbor of $S'$ not in the neighbor of $S^*_t$ remains the same from $k$ to $t+1$ round.
    
    With equation (\ref{eq10}) and (\ref{eq33}), we can conclude that $W(G_t)-W(G_{t+1})\geq \frac{1}{|\mathcal{M}|^2-|\mathcal{M}|}$ in case 2.\\	\\
In case 3, since every elements in $S^*_t$ belongs to $S^*_t$, $f_t(S^*_{t+1}) \geq f_{t+1}(S^*_{t+1})$.
	
	Given that $\dfrac{1}{|S^*_t|}<\dfrac{1}{|S^*_{t+1}|}$ and $f_t(S^*_t)-f_t(S^*_{t+1})>0$ by definition, $|S^*_t|-|S^*_{t+1}|+|N_t(S^*_t)|-|N_t(S^*_{t+1})|\geq 1$.
	
	 With the knowledge that $S^*_t$ and $S^*_t$ are non-empty set, it is obvious that $f_t(S^*_t)-f_t(S^*_{t+1})$ is lower-bounded by $\tfrac{1}{2}$. Therefore, we can conclude that $W(G)_t-W(G_{t+1})=f_t(S^*_t)-f_{t+1}(S^*_{t+1}) \geq f_t(S^*_t)-f_t(S^*_{t+1}) \geq \frac{1}{2}$ in case 3.\\ \\
	In case 4, define $S'=S^*_{t+1} \setminus S^*_t$, $S''=S^*_t \setminus S^*_{t+1}$, and $T=S^*_t \cap S^*_{t+1}$.
	\begin{eqnarray}
    && W(G_t)-W(G_{t+1}) \\
	&=&f_t(S^*_t)-f_{t+1}(S^*_{t+1}) \\
	&=& |T|+|S''|-|N_t(T)|-|N_t(S'')\setminus N_t(T)|+\frac{1}{|S^*_t|} - |T|-|S'|+|N_{t+1}(S^*_{t+1})|-\frac{1}{|S^*_{t+1}|} \\
	&=& |S''|-|N_t(S'')\setminus N_t(T)|+\frac{1}{|S^*_t|}-\frac{1}{|S^*_{t+1}|} +|N_{t+1}(S^*_{t+1})|-|N_t(T)|-|S'| \\
	&\geq & 1+\frac{1}{|S^*_t|}-\frac{1}{|S^*_{t+1}|}+|N_{t+1}(S^*_{t+1})| -|N_t(T)|-|S'| \\
	&\geq & 1+\frac{1}{|\mathcal{M}|}-1+|N_{t+1}(T))|-|N_t(T)|+|N_{t+1}(S')\setminus N_{t+1}(T)|-|S'|  \\
	&= & \frac{1}{|\mathcal{M}|}+|N_{t+1}(T))|-|N_t(T)| +|N_{t+1}(S')\setminus N_{t+1}(T)|-|S'| \label{eq11}\\
	&\geq & \frac{1}{|\mathcal{M}|}+|N_{t+1}(S')\setminus N_{k}(T)|-|S'| \geq  \frac{1}{|\mathcal{M}|} \label{eq12}\\
	\end{eqnarray}

	Most of the equations in case 4 are straight-forward except from (\ref{eq11}) to (\ref{eq12}).
	
	(\ref{eq11}) to (\ref{eq12}) is true because of the following inequalities:
	
	\begin{eqnarray}
	&&|N_{t+1}(S')\setminus N_{t+1}(T)|+|N_{t+1}(T))|-|N_t(T)| \nonumber \\
	&=&|N_{t+1}(S')\setminus N_t(T)|+|N_{t+1}(T))|-|N_t(T)|-| \{N_{t+1}(S') \cap N_{t+1}(T) \}\setminus N_t(T)|  \\
	&\geq & |N_{t+1}(S')\setminus N_t(T)|-|  N_{t+1}(T)\setminus N_t(T)| +|N_{t+1}(T))|-|N_t(T)| \\
	&=& |N_{t+1}(S')\setminus N_t(T)|
	\end{eqnarray}
	Combine these four cases, we know that
	\begin{eqnarray}
	W(G_t)-W(G_{t+1}) \geq \dfrac{1}{ |\mathcal{M}|^2-|\mathcal{M}|}.
	\end{eqnarray}
	
	Therefore, we can conclude that the proposed algorithm terminates in finite rounds. (At most $|\mathcal{M}|^3$ rounds because $W(G)<|\mathcal{M}|$ and minimum decrement is greater than $\tfrac{1}{|\mathcal{M}|^3}$).
\subsection{Proof of Lemma \ref{lem:colorMSS}}
First, we want to show that suppose $\mathcal{M}_g \cup\mathcal{M}_b$ is a constricted good set, we can not include any good colored red to increase the skewness of the set. For any set of good colored red $S_r$ being added to $\mathcal{M}_g \cup\mathcal{M}_b$, at least the same size of buyers being matched pairs of those red goods are join to $N(\mathcal{M}_g \cup\mathcal{M}_b \cup S_r)$. Since we know that $|N(S_r)\setminus N(\mathcal{M}_g \cup\mathcal{M}_b)|\geq |S_r|$, including any set of red goods will decrease the skewness of the set $\mathcal{M}_g \cup\mathcal{M}_b$.

Then, we want to show remove any subset of goods $S\subseteq \mathcal{M}_g \cup\mathcal{M}_b$ will also reduce the skewness of the set. Clearly, removing any subset of blue good will not reduce $N(\mathcal{M}_g \cup\mathcal{M}_b)$, hence blue goods are definitely included in the maximally skewed set. Therefore, we only need to consider the impact of removing set of green goods $S \subseteq \mathcal{M}_g$. Since for any $S \subseteq \mathcal{M}_g$, there has has at least one red edge connecting $S$ and $N(\mathcal{M}_g \cup\mathcal{M}_b\setminus S)$ according to the algorithm, Hence, $|N(\mathcal{M}_g \cup\mathcal{M}_b)|-|N(\mathcal{M}_g \cup\mathcal{M}_b\setminus S)|<|S|$, which implies that removing any $S \subset \mathcal{M}_g$ will only increase the skewness of the set. With these facts and the uniqueness property of the maximally skewed set, we can conclude that $\mathcal{M}_g \cup\mathcal{M}_b$ is the maximally skewed set.
\subsection{Proof of Lemma \ref{lem:alg2}}\label{pf:lemalg2}
	 We will prove Algorithm \ref{alg:alg2} always return the most skewed set by contradiction.\\
     Since every untraversed good upon termination can be matched with an untraversed buyer without repetition. Therefore, adding any set of runtraversed goods $S_U$ to the set $S$ will always reduce the skewness of $S$ (because the increase of cardinality in neighbor of $S$, $|N(S_U)|-|N(S)|$ is always greater than or equal to the increase of cardinality of $S$, $|S_U|-|S|$). Hence, the most skewed set will never contain any untraversed good.
     Suppose there exists a set $S'$ is the most skewed set, with a higher skewness than the set $S$ return by Algorithm \ref{alg:alg2}, $S'$ must be a subset of $S$ because there's no untraversed good in the most skewed set and $f(S')>f(S)$. Let $S^*=S\setminus S'$, $f(S')>f(S)$ and  $S' \subset S$ implies $|N(S^*)\setminus N(S')|-|S^*|>0$. If this happens, there must be a non-empty set of matching pairs match nodes from $S^*$ to $N(S^*)\setminus N(S')$ without repetition. Since $N(S^*)\setminus N(S')$ is not in $N(S')$ under the directed graph, nodes in $S^*$ will not be traversed in the algorithm contradicts that $S^* \subset S$.
     

\subsection{Proof of Theorem \ref{lem4.1}}\label{pf:lem4.1}
First, let's begin with the proof of the first half statement of the theorem, which is a variational characterization of MCPs.

($\Rightarrow$) This direction is obvious, otherwise the MCP is not the maximum by definition.

($\Leftarrow$) Recall that $\mathcal{M}$ is the set of goods. Suppose there exists an MCP $P^1$ satisfying the conditions but it is not the maximum MCP. Then there must exist a set of goods $S_1$ such that for all $i \in S_1$, $P^1_i<P^*_i$, and for all $i \in \mathcal{M}-S_1$, $P^1_i \geq P^*_i$; $\mathcal{M}-S_1$ can be an empty set.

Let $P^2_i= P^1_i$ for all $i \in \mathcal{M}-S_1$, and $P^2_i= P^*_i$ for all $i \in S_1$.

We will verify that $P^2$ is an MCP. WLOG, we can assume $P^1$ and $P^*$ have the same allocation; this is true as every MCP supports all efficient matchings. Then, consider any buyer who is assigned a good in $S_1$ under $P^*$. When the price vector changes from $P^*$ to $P^2$, the buyer has no profitable deviation from his/her assigned good because $P^2_i \geq P^*_i$ for all $i$ and $P^2_j= P^*_j$ for all $j \in S_1$. Similarly, when the price vector changes from $P^1$ to $P^2$, the buyer who is assigned a good in $\mathcal{M} \setminus S_1$ under $P^1$ has no profitable deviation because $P^2_i=P^1_i$ for all $i \in \mathcal{M} \setminus S_1$ and $P^2_j>P^1_j$ for all $j \in S_1$. Finally, since $P^1$ and $P^*$ have the same allocation, no buyer will deviate if we assign this allocation to buyers under $P^2$. Since all buyers have non-negative surpluses under $P^2$, it follows that $P^2$ is an MCP.

Given $P^1$, there exists a set of goods $S_1$ whose price we can increase and still get market clearing because both $P^1$ and $P^2$ are MCPs. This contradicts the assumption that $P^1$ satisfies the stated conditions, and the proof of the variational characterization follows.

For the second half part, which is a combinatorial characterization of MCPs, let's try to prove the statement using the result of the variational characterization.
	Given that we cannot increase the price for any subset of goods, it implies that for any subset of goods $S$, the set of corresponding matched buyer, either there exists at least one buyer has an edge connected to a good not in this subset or there exists at least one buyer with surplus zero. In the former case, it is obvious that $B<N(B)$ for such corresponding buyer set $B$. In the later case, $B\leq N(B)$ and $D\in N^D(B)$ guarantees $B < N^D(B)$.
    
    For the opposite direction, if every set of buyer with  $B < N^D(B)$ under the current MCP, increasing the price for any set of good will make at least one corresponding buyer deviate from the matched good and cause no perfect matching. Therefore, the condition in variational characterization holds if and only if the condition in the combinatorial characterization holds.
\subsection{Proof of Theorem \ref{thm2}}\label{pf:thm2}

As mentioned earlier, showing the algorithm returns the maximum MCP is equivalent to show that when the algorithm terminates, the bipartite structure guarantees that we can not increase the price of any subset of good with the result in Theorem \ref{lem4.1}. It implies that after adding a dummy good  with zero price, the support of any subset of goods has a size less than the size of the subset of goods, i.e., $|B|< |N^D_T(B)|$. Therefore, the proof of the theorem can be transformed to prove that the algorithm satisfies $|B|< |N^D_T(B)|$ for any non-empty set of buyer $B$ on termination. Let's start the proof with several claims.

\begin{claim} \label{cl1}
 For any subset of buyer $B$, $B\neq \emptyset$, $|B|\leq |N^D_T(B)|$, where $T$ is the terminating time of our algorithm.
\end{claim}

 Since our algorithm returns an MCP vector, if $|B|> |N^D_T(B)|$, there does not exist a perfect matching because of the existence of constricted buyer set.

\begin{claim} \label{cl2}
There does not exists a subset of buyer $B\subseteq \mathcal{B}$, $B\neq \emptyset$, $|B|= |N^D_T(B)|$ and $D\in N^D_T(B)$, where $T$ is the terminating time of our algorithm. \end{claim}
$|B|= |N^D_T(B)|$ and $D\in N^D_T(B)$, imply $|B|> |N_T(B)|$. The preference graph has constricted sets and has no perfect matching.
\begin{claim} \label{cl3}
There does not exists a subset of buyer $B\subseteq \mathcal{B}$, $B\neq \emptyset$, $|B|= |N_T(B)|$, where $T$ is the terminating time of our algorithm.
\end{claim}
\begin{proof}
    With Claim \ref{cl1}, \ref{cl2}, it is equivalent to show $|B|<|N_T(B)|$.

    At time $t$, we denote the maximum non-negative surplus of buyer $b$ by $u^*_t(b)$ and the most skewed set by $S^*_t$; and $\mathscr{B}^s_t$ is the set of buyers with positive surplus at time $t$, i.e., $\mathscr{B}^s_t=\{b|u^*_t(b)>0, b\in\mathcal{B}\}$.

    It is obvious that $\mathscr{B}^s_i \subseteq \mathscr{B}^s_j$ for all $i<j$. Then, we want to prove the claim by mathematical induction.

    Prior to the proof, we need to introduce another claim.
        \begin{claim} \label{cl4}
        $\forall B\subseteq N_t(S^*_t), 0\leq t<T$, $|B|<|N_t(B) \cap S^*_t)|=|N_{t+1}(B)|$
    \end{claim}
       \begin{proof}
           If the left inequality does not hold, we can remove all the goods contained in $N_t(B) \cap S^*_t$ from $S^*_t$ to get a more skewed set, which violates that $S^*_t$ is the most skewed one.

           Since we reduce the price in $S^*_t$, buyers in $B$ will not prefer any good outside of $S^*_t$. Therefore, $N_t(B) \cap S^*_t$ and $N_{t+1}(B)$ are identical. Hence $|N_t(B) \cap S^*_t|=|N_{t+1}(B)|$ is absolutely true.
       \end{proof}

    Now, we can prove Claim \ref{cl3} by induction.

    At t=0, $\mathscr{B}^s_0=\emptyset$.

    At t=1, $\mathscr{B}^s_1=\mathcal{N_0}(S^*_0)$. With Claim \ref{cl4}, $|B|<|N_1(B)|~~\forall B\subseteq \mathscr{B}^s_1$ and $B \neq \emptyset$.

    At a finite time $t$, suppose for all $B\subseteq \mathscr{B}^s_t$, $B \neq \emptyset$ satisfy $|B|<|N_t(B)|$, consider at time $t+1$:

    Since $\mathscr{B}^s_t \subseteq \mathscr{B}^s_{t+1}$, $\mathscr{B}^s_{t+1}$ contains three disjoint components:
    \begin{equation}
        \mathscr{B}^s_{t+1}= \{\mathscr{B}^s_t \cap {N_t(S^*_t)}^c \} \cup \{\mathscr{B}^s_t \cap N_t(S^*_t) \} \cup
        \{{\mathscr{B}^s_t}^c \cap N_t(S^*_t) \}
    \end{equation}

    Buyers in the first two parts are originally with positive utilities. Buyers in the last part have zero utilities at time $t$ but have positive utilities at time $t+1$.

    Consider the subset of buyers $B_{\alpha} \subseteq \{ \mathscr{B}^s_t \cap {N_t(S^*_t)}^c \}$. Since every $b \in B_{\alpha}$ does not prefer any good in $S^*_t$, the price reduction in $S^*_t$ will never remove any edges between $\{ \mathscr{B}^s_t \cap {N_t(S^*_t)}^c \}$ and $N_t(\{ \mathscr{B}^s_t \cap {N_t(S^*_t)}^c \})$. Therefore, for any non-empty set of buyers $B_{\alpha}$, $|B_{\alpha}|< |N_t(B_{\alpha})|\leq |N_{t+1}(B_{\alpha})|$.

    Then, consider the second and the third parts. Since
    $\{\mathscr{B}^s_t \cap N_t(S^*_t) \} \cup
        \{{\mathscr{B}^s_t}^c \cap N_t(S^*_t) \}=N_t(S^*_t)$, every non-empty set of buyers $B_{\beta} \subseteq N_t(S^*_t)$ satisfies $|B_{\beta}|<|N_{t+1}(B_{\beta})|$ by Claim \ref{cl4}.

    At the last step, consider $B_{\gamma}=B_{\gamma_1} \cup B_{\gamma_2}$, where $B_{\gamma_1} \subseteq \{\mathscr{B}^s_t \cap {N_t(S^*_t)}^c \}$, $B_{\gamma_2} \subseteq N_t(S^*_t)$ and $B_{\gamma_1},B_{\gamma_2} \neq \emptyset$.

    \begin{eqnarray}
        |N_{t+1}(B_{\gamma})|
        &=&|N_{t+1}(B_{\gamma_1}) \cup
        N_{t+1}(B_{\gamma_2})| \\
        &=& |N_{t+1}(B_{\gamma_1})|+|N_{t+1} (B_{\gamma_2})|-|N_{t+1}(B_{\gamma_1}) \cap N_{t+1}(B_{\gamma_2})| \\
        &=& |N_{t+1}(B_{\gamma_1})|-|N_{t+1}(B_{\gamma_1}) \cap N_{t+1}(B_{\gamma_2})|+|N_{t+1} (B_{\gamma_2})| \\
        &=& |N_{t+1}(B_{\gamma_1}) \cap {N_{t+1}(B_{\gamma_2})}^c|+|N_{t+1} (B_{\gamma_2})|  \\
        &\geq& |N_{t+1}(B_{\gamma_1}) \cap {N_{t+1}(N_t(S^*_t))}^c|+|N_{t+1} (B_{\gamma_2})| \label{eqx1} \\
        &=& |N_{t+1}(B_{\gamma_1}) \cap {S^*_t}^c|+|N_{t+1} (B_{\gamma_2})| \label{eqx2}\\
        &=& |N_t(B_{\gamma_1})|+|N_{t+1} (B_{\gamma_2})| \label{eqx3}\\
        &>& |B_{\gamma_1})|+|B_{\gamma_2}|=|B_{\gamma}|
    \end{eqnarray}

    From (\ref{eqx1}) to (\ref{eqx2}) is true because after price reduction at time $t$, buyer belongs to the neighbor of constricted good set $S^*_t$ will only prefer goods in $S^*_t$, therefore $N_{t+1}(N_t(S^*_t))=S^*_t$.

    From (\ref{eqx2}) to (\ref{eqx3}) is true because $N_t(B_{\gamma_1}) \cap S^*_t = \emptyset$ (by definition). Then, we have discussed before that the price reduction in $S^*_t$ will never remove any edges between $\{ \mathscr{B}^s_t \cap {N_t(S^*_t)}^c \}$ and $N_t(\{ \mathscr{B}^s_t \cap {N_t(S^*_t)}^c \})$.
    Since goods not in the most skewed set at time $t$ will not add new buyers to their neighbor, the equivalence between $N_{t+1}(B_{\gamma_1}) \cap {S^*_t}^c$ and $N_t(B_{\gamma_1})$ holds.

    Therefore, for every non empty set of buyers $B\subseteq \mathscr{B}^s_{t+1}$, $|B|<|N_{t+1}(B)|$.

    The mathematical induction works for all $k \in \mathbb{N}$. Since our algorithm terminates in finite round, $|B|< |N_T(B)|$, Q.E.D.
\end{proof}
With Claim \ref{cl1}, \ref{cl2}, \ref{cl3}, the skew-aided algorithm satisfies $|B|<|N^D_t(B)|$ for any non-empty set of buyer $B$ on termination.

\subsection{Proof of Theorem \ref{thm1}}\label{pf:thm1}
	First, we know from Section 5.2 that the preference graph changes at most $m^2$ times in Algorithm \ref{alg:alg1}.
	
	Second, we analyze the initial step first. In Algorithm \ref{alg:alg2}, the Hopcroft-Karp algorithm runs in time $O(|\mathcal{M}|^{2.5})$, and the complexity of the BFS algorithm is $O(m+|\text{number of edges}|)$. Since the number of directed edges is upper-bounded by $m^2$, the complexity of Algorithm \ref{alg:alg2} is $O(m^{2.5}+m^2)=O(m^{2.5})$. Then, for each $a \in A$ colored blue, we run at most twice of BFS/DFS algorithm, which has complexity 
upper-bounded by $O(m+|\text{number of edges}|)\geq O(m^2)$. After that, we do things similar to find the maximum matching if we can add at least one blue buyer to the maximum matching, and run at most twice BFS/DFS. Therefore, the complexity should be upper-bounded by $O(m^{2.5}+2\times m^2)=O(m^{2.5})$. If we know that there is no blue buyers going to be added in this round, we do not need to find a new maximum matching. Therefore, the total complexity of the algorithm attaining the maximum MCP will be upper-bounded by the \{(complexity of updating process not recoloring blue buyers)$+$(complexity of computing price reduction)\}$\times$ (convergence rate of the preference graph)$+$(complexity of updating process increasing maximum matched pairs)$\times$ (maximum number of blue buyers)$=O((m^2+m^2)\times m^2+m^{2.5}\times m)=O(m^4)$ 

\subsection{Proof of Lemma \ref{lem:multicon}}
     Suppose the picked set which is not the most skewed one, this set must contain a pair of good-buyer $(x,y)$ not in the maximally skewed set. 
         Since we know that algorithm terminates at round $T$, reducing the price of the good $x$ make the price vector lower than the maximum MCP at round $T$. (Because if we pick another set excluding good $x$, reducing the same amount of price as we did in this algorithm will also terminate at $T$.) Therefore, we know that the constricted good set we pick should not include any good not in the maximally skewed set.
    Now, we want to claim that the picked set is a subset of the maximally skewed set will never happen. If such a set $S$ exists and can terminate the algorithm at round $T$, then the most skewed set $S^*$ contains a set of goods $S^*\setminus S$ is already perfectly matched to a set of buyers outside $N(S)$, and the skewness of $S$ must be greater than $S^*$, which contradicts that $S^*$ is the most skewed set. 
	Because of the above claims, we can conclude that any choice of constricted good set which is not the maximally skewed set at round $T$ will never return the maximum MCP at round $T$.
	\subsection{Proof of Theorem \ref{thm:externality}}\label{pf:thm:externality}
	First, it is obvious that the current market adding a duplicate pair of buyer buyer $i$ and its matched good $j$ is still market clearing and the social welfare will be the  $\text{(the current social welfare)}+U^*_i+P_j$. Hence, what we need to show is the current market adding a duplicate buyer buyer $i$ has the social welfare:     $\text{(the current social welfare)}+U^*_i$. Since adding a dummy good will not change the social welfare, we can transfer our problem to prove that the current market adding a duplicate buyer buyer $i$ and a dummy good has the social welfare $\text{(the current social welfare)}+U^*_i+0$. It is equivalent to show that the current market adding a duplicate buyer buyer $i$ and a dummy good is still market clearing. Using the combinatorial characterization of the maximum MCP, we know that $|B|<|N^D(B)|$ for any subset of B in the current market. Denote the buyer set of current market as $\mathcal{B}$ and the duplicate buyer $i$ as $\hat{i}$ , what we want to show is that $|B|\leq |N^D(B)|$ for any $B \in \{\mathcal{B}\cup\hat{i}\}$.

    If at least one of $i, \hat{i}$ not in $B$, it is straightforward that $|B|<|N^D(B)|$.

    If both $i, \hat{i}\in B$, $|B|=|B\setminus \hat{i}|+1 \leq |N^D(B\setminus \hat{i})-1|+1=|N^D(B\setminus \hat{i})|\leq |N^D(B)|)$.

    Using the Hall's marriage theorem, the inequality guaranteed that the current market adding a duplicate buyer buyer $i$ and a dummy good is market clearing, and the proof stands here.

\subsection{Verification of Asymmetric BNE in a $3\times 3$ matching market}\label{asymBNE}
    Given the valuation matrix described in Table \ref{tbl:1}, we want to verify that if the strategy profile $\{\beta_{Alice}(w,0,0), \beta_{Bob}(x,0,0.5),\beta_{Carol}(y,z,2)\}=\{(\max \{1-\epsilon, \frac{w}{2}\},0,0), (\max \{1, \frac{2x-1}{4}\},0,0),(0,\epsilon,0)\}$, where $\epsilon$ is an infinitesimal is a BNE.
    
    First, consider Alice's best response according to Bob's strategy $(\max \{1, \frac{2x-1}{4}\},0,0)$ and Carol's strategy $(0,\epsilon,0)$. When Alice has valuation between $[0,1]$ on the listing ads, any bidding strategy $(b,0,0)$ is a best response for all $b<1$ because Alice will always get the Pop-ups at price zero. Then, when Alice has valuation $w$ between $(2,3]$ on the listing ads. The bidding function maximizes Alice's expected payoff is
    \begin{eqnarray}
    &&\max_b \int_{1.5}^{3.5} (w-b)\mathbf{1}_{\{b>\max \{1, (2x-1)/4\}\}} f_x(x) dx \\
    &=& \max_b \Big\{ \frac{2}{3}(w-b)+ \frac{1}{3}\int_{2.5}^{3.5} (w-b)\mathbf{1}_{\{b>\max \{1, (2x-1)/4\}\}} dx \Big\} \\
    &=& \max_b \frac{2}{3}(w-b)+ \frac{1}{3}(w-b)\times 2(b-1)
    \end{eqnarray}
    Now, it is easy to solve the optimal bid of Alice is
     \begin{eqnarray}
    &&\text{arg}\max_b \frac{2}{3}(w-b)+ \frac{1}{3}(w-b)\times 2(b-1) \\
    &=& \text{arg}\max_b (w-b)+(w-b)(b-1)=\frac{w}{2}
    \end{eqnarray}
    Since $\frac{w}{2}<1$ for all $w \in [0,1)$ and $\frac{w}{2}>1$ for all $w \in (2,3]$ , we can conclude that $(\max \{1-\epsilon, \frac{w}{2}\},0,0)$ is a best response of Alice under $\{\beta_{Bob}(x,0,0.5),\beta_{Carol}(y,z,2)\}=\{(\max \{1, \frac{2x-1}{4}\},0,0),(0,\epsilon,0)\}$.
    
    Second, we use the similar technique to get Bob's best response. Given the strategy of Alice and Carol as mentioned, Bob can always get the Pop-ups with price $0$. Therefore, the bidding function maximizes Bob's expected payoff is 
    \begin{eqnarray}
    &&\max_b 0.5+\int_{0}^{3} (w-b-0.5)\mathbf{1}_{\{b>\max \{1, \frac{w}{2}\}\}} f_w(w) dw \\
    &=& 0.5+\max_b \Big\{ \frac{2}{3}(w-b-0.5)\mathbf{1}_{\{b\geq 1 \}}+ \frac{1}{3}\int_{2}^{3} (w-b-0.5)\mathbf{1}_{\{b>\max \{1, \frac{w}{2}\}\}} dw \Big\} \\
    &=& 0.5+\max_b \frac{2}{3}(w-b-0.5)\mathbf{1}_{\{b\geq 1 \}}+ \frac{2}{3}(w-b-0.5)(b-1)
    \end{eqnarray}
   The optimal bid of Bob is            
  	\begin{eqnarray}
    &&\text{arg}\max_b \frac{2}{3}(w-b-0.5)\mathbf{1}_{\{b\geq 1 \}}+ \frac{2}{3}(w-b-0.5)(b-1) \\
    &=& \text{arg}\max_b (w-b-0.5)(b-1+\mathbf{1}_{\{b\geq 1 \}})=\max \{1, \frac{2x-1}{4}\}
    \end{eqnarray}
    Last, we have to verify the Carol's best response given
    $\{\beta_{Alice}(w,0,0), \beta_{Bob}(x,0,0.5)\}=\{(\max \{1-\epsilon, \frac{w}{2}\},0,0), (\max \{1, \frac{2x-1}{4}\},0,0)\}$. To against any tie-breaking rule not in favor of Carol, the strategy $(0,\epsilon,0)$ is the minimum bid to ensure getting the sidebar ads. Now, we complete the verification that $\{\beta_{Alice}(w,0,0), \beta_{Bob}(x,0,0.5),\beta_{Carol}(y,z,2)\}=\{(\max \{1-\epsilon, \frac{w}{2}\},0,0), (\max \{1, \frac{2x-1}{4}\},0,0),(0,\epsilon,0)\}$ is an (asymmetric) BNE in this $3 \times 3$ matching market.   
    
\subsection{Proof of Lemma \ref{lem:mono}}\label{pf:lem:mono}
     We prove the monotonicity by contradiction.

     Suppose there exists two weight $x>y$ such that $\beta(x)<\beta(y)$, where $\beta (\cdot)$ is the optimal bidding function. It guarantees that the expected surplus of bidding $\beta(x)$ is never worse than bidding $\beta(y)$ given the private weight $x$. (Otherwise it is not the optimal bidding function.) Similarly, it also has to satisfy that the expected surplus of bidding $\beta(y)$ will be never worse than bidding $\beta(x)$ given the private weight $y$. Mathematically, the following two inequality should hold.
       \begin{eqnarray} \label{eqn:x>y}
         \int_0^1(c_1x-[(c_1-c_2)\beta(x) +c_2\beta(u)])f_{w}(u)\mathbf{1}_{\{\beta(x)>\beta(u)\}}
         du +\int_0^1 c_2(x-\beta(x))f_{w}(u)\mathbf{1}_{\{\beta(x)\leq \beta(u)\}}du \nonumber  \\
         \geq          \int_0^1(c_1x-[(c_1-c_2)\beta(y) +c_2\beta(u)])f_{w}(u)\mathbf{1}_{\{\beta(y)>\beta(u)\}}
         du +\int_0^1 c_2(x-\beta(y))f_{w}(u)\mathbf{1}_{\{\beta(y)\leq \beta(u)\}}du 
        \end{eqnarray}
       \begin{eqnarray} \label{eqn:y>x}
         \int_0^1(c_1y-[(c_1-c_2)\beta(y) +c_2\beta(u)])f_{w}(u)\mathbf{1}_{\{\beta(y)>\beta(u)\}}
         du +\int_0^1 c_2(y-\beta(y))f_{w}(u)\mathbf{1}_{\{\beta(y)\leq \beta(u)\}}du \nonumber  \\
         \geq          \int_0^1(c_1y-[(c_1-c_2)\beta(x) +c_2\beta(u)])f_{w}(u)\mathbf{1}_{\{\beta(x)>\beta(u)\}}
         du +\int_0^1 c_2(y-\beta(x))f_{w}(u)\mathbf{1}_{\{\beta(x)\leq \beta(u)\}}du 
        \end{eqnarray}
        
        Now, let us sum up and summarize the inequalities (\ref{eqn:x>y}), (\ref{eqn:y>x}). 
        \begin{eqnarray}
\int_0^1c_1(x-y)f_{w}(u)\mathbf{1}_{\{\beta(x)>\beta(u)\}}
         du +\int_0^1 c_2(x-y)f_{w}(u)\mathbf{1}_{\{\beta(x)\leq \beta(u)\}}du \nonumber  \\
         \geq          \int_0^1c_1(x-y)f_{w}(u)\mathbf{1}_{\{\beta(y)>\beta(u)\}}
         du +\int_0^1 c_2(x-y)f_{w}(u)\mathbf{1}_{\{\beta(y)\leq \beta(u)\}}du 
        \end{eqnarray}
Then, the inequality can be further simplified to the following:
        \begin{eqnarray} \label{eqn:xy>yx}
		(x-y) \int_0^1(c_1-c_2)f_{w}(u)(\mathbf{1}_{\{\beta(x)>\beta(u)\}}-\mathbf{1}_{\{\beta(y)>\beta(u)\}})du\geq  0
        \end{eqnarray}
        Since $x>y$ and $\beta(x)<\beta(y)$, $(x-y) \int_0^1(c_1-c_2)f_{w}(u)(\mathbf{1}_{\{\beta(x)>\beta(u)\}}-\mathbf{1}_{\{\beta(y)>\beta(u)\}})du<0$ holds given that the private weight $w$ is uniformly distributed from [0,1], which contradicts inequality \ref{eqn:xy>yx} and the proof stands here.

\subsection{Proof of Corollary \ref{cor:mono}}\label{pf:cor:mono}
Using the same technique as we used in the proof of Lemma \ref{lem:mono}, we can use induction to generalized to a finite-slots sponsored search markets.

     Consider the advertisers' private weight are symmetrical distributed from a distribution $D$. Suppose there exists two weight $x>y$ such that $\beta(x)<\beta(y)$, where $\beta (\cdot)$ is the optimal bidding function. Suppose there are $n$ slots and let $u_i$ be other advertiser $i$'s private weight, assuming $\beta(u_i)\geq \beta(u_j)~\forall i>j$ without loss of generality. Define the space $A=\{(u_1,u_2,...,u_{n-1})|\beta(u_i)\geq \beta(u_j)~\forall i<j\}$, $\beta(u_0)=+\infty$, and $c_{n+1}=0$, the two inequalities that the optimal bidding function have to satisfy are as follows:
            \begin{eqnarray} 
         &&\sum_{i=1}^{n-1} \tbinom{n-1}{k} \oint_A \bigg\{c_ix-\beta(x)(c_i-c_{i+1})-\sum_{j=i+1}^n \beta(u_{j-1})(c_j-c_{j+1})\bigg\}\mathbf{1}_{\{\beta(u_{i-1})\geq \beta(x)>\beta(u_i)\}} \prod_{k=i}^{n-1}f_D(u_k)du_k \nonumber \\
         &\geq& \sum_{i=1}^{n-1} \tbinom{n-1}{k} \oint_A \bigg\{c_ix-\beta(y)(c_i-c_{i+1})-\sum_{j=i+1}^n \beta(u_{j-1})(c_j-c_{j+1})\bigg\}\mathbf{1}_{\{\beta(u_{i-1})\geq \beta(y)>\beta(u_i)\}} \prod_{k=i}^{n-1}f_D(u_k)du_k \nonumber
        \end{eqnarray}
        
            \begin{eqnarray}
         &&\sum_{i=1}^{n-1} \tbinom{n-1}{k} \oint_A \bigg\{c_iy-\beta(y)(c_i-c_{i+1})-\sum_{j=i+1}^n \beta(u_{j-1})(c_j-c_{j+1})\bigg\}\mathbf{1}_{\{\beta(u_{i-1})\geq \beta(y)>\beta(u_i)\}} \prod_{k=i}^{n-1}f_D(u_k)du_k \nonumber \\
         &\geq& \sum_{i=1}^{n-1} \tbinom{n-1}{k} \oint_A \bigg\{c_iy-\beta(x)(c_i-c_{i+1})-\sum_{j=i+1}^n \beta(u_{j-1})(c_j-c_{j+1})\bigg\}\mathbf{1}_{\{\beta(u_{i-1})\geq \beta(x)>\beta(u_i)\}} \prod_{k=i}^{n-1}f_D(u_k)du_k \nonumber
        \end{eqnarray}
        Now, we have to sum up the two inequalities above and cancel terms exist in both sides.
         \begin{eqnarray} \label{eqn:mono}
         &&\sum_{i=1}^{n-1} \tbinom{n-1}{k} \oint_A c_i(x-y)\mathbf{1}_{\{\beta(u_{i-1})\geq \beta(x)>\beta(u_i)\}} \prod_{k=i}^{n-1}f_D(u_k)du_k \nonumber \\
         &\geq& \sum_{i=1}^{n-1} \tbinom{n-1}{k} \oint_A c_i(y-x)\mathbf{1}_{\{\beta(u_{i-1})\geq \beta(y)>\beta(u_i)\}} \prod_{k=i}^{n-1}f_D(u_k)du_k 
        \end{eqnarray}
        We can further simplify the inequality (\ref{eqn:mono}) to be the following inequality:
         \begin{eqnarray} 
        (x-y)\sum_{i=1}^{n-1} \tbinom{n-1}{k} \oint_A c_i\{\mathbf{1}_{\{\beta(u_{i-1})\geq \beta(x)>\beta(u_i)\}}-\mathbf{1}_{\{\beta(u_{i-1})\geq \beta(y)>\beta(u_i)\}}\}\prod_{k=i}^{n-1}f_D(u_k)du_k \geq 0
        \end{eqnarray} 
        Given that $c_i\geq c_j$ for all $i<j$, $x>y$ and $\beta(x)<\beta(y)$, the above inequality is always $\leq 0$ and the equality holds when $F_D(x)-F_D(y)=0$. Therefore, we can claim that the optimal bidding function in sponsored search markets with symmetric advertisers are monotonic increasing with the advertiser's private weight.
            
\subsection{Proof of Lemma \ref{lem_2g1}}\label{sec:lem_2g1}

Consider buyer i's bid $b_{i1}$, $b_{i2}$ on good 1 and 2. WLOG, suppose $b_{11},b_{12}\neq 0$. If $b_{11}+b_{22}\geq b_{12}+b_{21}$ (the scenario that player 1 will win good 1). The price of good 1 is $P_1=b_{11}-(b_{12}-b_{22})\mathbf{1}_{\{b_{11}\geq b_{21},b_{12}\geq b_{22}\}}$. Similarly, if $b_{11}+b_{22}< b_{12}+b_{21}$ (the scenario that player 1 will win good 2). The price of good 2 at this time is $P_1=b_{12}-(b_{11}-b_{21})\mathbf{1}_{\{b_{11}\geq b_{21},b_{12}\geq b_{22}\}}$.

The buyer 1's surplus is
\begin{eqnarray}
&&u_1(b_{11},b_{12},b_{21},b_{22})\nonumber \\
&=&(v_{11}-b_{11})\mathbf{1}_{\{b_{11}+b_{22}\geq b_{12}+b_{21}\}} + (v_{12}-b_{12})(1-\mathbf{1}_{\{b_{11}+b_{22}\geq b_{12}+b_{21}\}})\nonumber\\
&&+\mathbf{1}_{\{b_{11}\geq b_{21},b_{12}\geq b_{22}\}}\big[(b_{12}-b_{22})\mathbf{1}_{\{b_{11}+b_{22}\geq b_{12}+b_{21}\}}+(b_{11}-b_{21})(1-\mathbf{1}_{\{b_{11}+b_{22}\geq b_{12}+b_{21}\}})\big] \nonumber
\end{eqnarray}
Now, define $c=\min \{b_{11},b_{12}\}$ and consider another bidding strategy $(b^*_{11} , b^*_{12})$ that $(b_{11}^*=b_{11}-c$, $(b_{12}^*=b_{12}-c$. Since reducing same amount of bids on both goods reduces the same amount on $b_{11}+b_{22}$ and $b_{12}+b_{21}$, this new bidding strategy will not change the probability of buyer 1 to win good 1 or 2. Therefore, we can calculate the difference of surplus between these two bidding strategy:
\begin{eqnarray}
&&u_1(b^*_{11},b^*_{12},b_{21},b_{22})-u_1(b_{11},b_{12},b_{21},b_{22}) \\
&=&(b_{11}-b^*_{11})\mathbf{1}_{\{b_{11}+b_{22}\geq b_{12}+b_{21}\}} + (b_{12}-b^*_{12})(1-\mathbf{1}_{\{b_{11}+b_{22}\geq b_{12}+b_{21}\}})\nonumber\\
&&-\mathbf{1}_{\{b_{11}\geq b_{21},b_{12}\geq b_{22}\}}\big[(b_{12}-b_{22})\mathbf{1}_{\{b_{11}+b_{22}\geq b_{12}+b_{21}\}}+(b_{11}-b_{21})(1-\mathbf{1}_{\{b_{11}+b_{22}\geq b_{12}+b_{21}\}})\big] \nonumber \\
&=& c[\mathbf{1}_{\{b_{11}+b_{22}\geq b_{12}+b_{21}\}} +(1-\mathbf{1}_{\{b_{11}+b_{22}\geq b_{12}+b_{21}\}})] \nonumber \\
&&-\mathbf{1}_{\{b_{11}\geq b_{21},b_{12}\geq b_{22}\}}\big[(b_{12}-b_{22})\mathbf{1}_{\{b_{11}+b_{22}\geq b_{12}+b_{21}\}}+(b_{11}-b_{21})(1-\mathbf{1}_{\{b_{11}+b_{22}\geq b_{12}+b_{21}\}})\big] \label{eqn12} \\
&\geq & c- \mathbf{1}_{\{b_{11}\geq b_{21},b_{12}\geq b_{22}\}}\big[c\mathbf{1}_{\{b_{11}+b_{22}\geq b_{12}+b_{21}\}}+c(1-\mathbf{1}_{\{b_{11}+b_{22}\geq b_{12}+b_{21}\}})\big] \label{eqn13}  \\
&=& c- c\mathbf{1}_{\{b_{11}\geq b_{21},b_{12}\geq b_{22}\}} \geq 0 \label{eqn14}
\end{eqnarray}
The most critical part is from (\ref{eqn12}) to (\ref{eqn13}). When $b_{11}+b_{22}\geq b_{12}+b_{21}, b_{11}\geq b_{21}$, and $b_{12}\geq b_{22}$, they imply that $b_{12}\geq b_{22}\leq \min\{b_{11},b_{12}\}=c$. Similarly, $b_{11}+b_{22}\leq b_{12}+b_{21}, b_{11}\geq b_{21}$, and $b_{12}\geq b_{22}$ imply $b_{11}\geq b_{21}\leq \min\{b_{11},b_{12}\}=c$. With Equation (\ref{eqn14}), we can conclude that any strategy placing non-zero bids on both goods is a weakly-dominated strategy.

\subsection{Proof of Lemma \ref{lem_2g2}}\label{sec:lem_2g2}
If a strategy bid $b_{i1}, b_{i2}$ satisfying $b_{i1}<b_{i2}$ when $v_{i1}>v_{i2}$, consider another strategy bid $b^*_{i1}, b^*_{i2}$ as follows:
\begin{eqnarray}
&&b^*_{i1}=b_{i1} \nonumber\\
&&b^*_{i2}=
\begin{cases} b_{i1}~~~~ \text{if~}  v_{i1}>v_{i2} \\
              b_{i2}~~~~ \text{otherwise}
\end{cases}
\end{eqnarray}
With Lemma \ref{lem_2g1}, we can assume $b_{i1}=0$.
Used the similar technique to compare the buyer's surplus when $v_{i1}>v_{i2}$.
\begin{eqnarray}
&&u_i(b^*_{i1},b^*_{i2},b_{-i1},b_{-i2})-u_i(b_{i1},b_{i2},b_{-i1},b_{-i2})
= u_i(0,0,b_{-i1},b_{-i2})-u_i(0,b_{i2},b_{-i1},b_{-i2}) \nonumber \\
&=&v_{i1}\mathbf{1}_{\{b_{-i2}\geq b_{-i1}\}} + v_{i2}(1-\mathbf{1}_{\{b_{-i2}\geq b_{-i1}\}})-
v_{i1}\mathbf{1}_{\{b_{-i2}\geq b_{i2}+ b_{-i1}\}} - (v_{i2}-b_{i2})(1-\mathbf{1}_{\{b_{-i2}\geq  b_{i2}+b_{-i1}\}})\nonumber \\
&=& b_{i2}(1-\mathbf{1}_{\{b_{-i2}\geq  b_{i2}+b_{-i1}\}})+(v_{i1}-v_{i2})\mathbf{1}_{\{b_{-i1}+b_{i2}\geq b_{-i2}\geq b_{-i1}\}} \geq 0 \label{eqn16}
\end{eqnarray}
Since equation (\ref{eqn16}) is non-negative. The lemma is proved.

\subsection{Proof of Lemma \ref{pro1}}\label{sec:pro1}
 Suppose $\beta(v_{ih},v_{il})=b^*>b'=\beta(v_{ih}',v_{il}')$ but $v_{ih}-v_{il}<v_{ih}'-v_{il}'$, consider another bidding function $\beta'$ exchanges the bid between these two pair of valuations and bids the same as $\beta$ otherwise.
 Suppose the original bidding function is optimal, the following two inequality holds:
 \begin{eqnarray}
 &&[(v_{ih}-\beta(v_{ih},v_{il}))\text{Pr(win higher valued good with bid $\beta(v_{ih},v_{il}$))} \nonumber \\
 &+&v_{il}\text{Pr(win lower valued good with bid $\beta(v_{ih},v_{il}$))}] \nonumber \\
 &-&(v_{ih}-\beta'(v_{ih},v_{il}))\text{Pr(win higher valued good with bid $\beta'(v_{ih},v_{il}$))} \nonumber \\
 &-&v_{il}\text{Pr(win lower valued good with bid $\beta'(v_{ih},v_{il}$))} \geq 0
 \end{eqnarray}
  \begin{eqnarray}
 &&[(v_{ih}'-\beta(v_{ih}',v_{il}'))\text{Pr(win higher valued good with bid $\beta(v_{ih}',v_{il}'$))} \nonumber \\
 &+&v_{il}'\text{Pr(win lower valued good with bid $\beta(v_{ih}',v_{il}'$))}] \nonumber \\
 &-&(v_{ih}'-\beta'(v_{ih}',v_{il}'))\text{Pr(win higher valued good with bid $\beta'(v_{ih}',v_{il}'$))} \nonumber \\
 &-&v_{il}'\text{Pr(win lower valued good with bid $\beta'(v_{ih}',v_{il}'$))} \geq 0
 \end{eqnarray}
 Sum up the above two inequality and we know that

 $\text{Pr(win higher valued good with bid $\beta'(v_{ih}',v_{il}'$))}=\text{Pr(win higher valued good with bid $\beta(v_{ih},v_{il}$))}$

 We will get
 \begin{eqnarray}
&& [(v_{ih}-v_{il})-(v_{ih}'-v_{il}')][\text{Pr(win higher valued good with bid $b^*$)} \nonumber\\
 &-&\text{Pr(win higher valued good with bid b')}]\geq 0
 \end{eqnarray}
 However, we know $v_{ih}-v_{il})<(v_{ih}'-v_{il}')$ and $b^*>b'$. there's a contradiction. Hence, the optimal bidding function of the higher good $\beta(v_{ih},v_{il})$ is monotonic to $v_{ih}-v_{il}$.        
\end{document}